\documentclass[a4paper,USenglish,cleveref,autoref]{lipics-v2021}

\usepackage{amsmath,amssymb}
\usepackage{xspace}
\usepackage{booktabs}
\usepackage{graphicx}
\usepackage{tikz}
\usetikzlibrary{arrows.meta,positioning,calc,shapes.geometric}
\usepackage[ruled,vlined,linesnumbered]{algorithm2e}

\newcommand{\nnip}{\textsf{nni+}\xspace}
\newcommand{\dn}{d_{\mathrm{nni+}}}
\newcommand{\de}{d_{\mathrm{nni+}_1}}
\newcommand{\enni}{\textsf{nni+}$_1$\xspace}
\newcommand{\dpc}{d_{\mathrm{PC}}}
\newcommand{\dcl}{d_{\mathrm{CL}}}
\newcommand{\ddep}{d_{\mathrm{dep}}}

\newcommand{\dnni}{d_{\mathrm{NNI}}}

\newcommand{\dgnni}{d_{\mathrm{gNNI}}}
\newcommand{\LB}{\mathrm{LB}}
\newcommand{\T}{\mathcal{T}}
\newcommand{\NP}{\textsf{NP}\xspace}
\newcommand{\pa}{\mathrm{p}}
\newcommand{\ch}{\mathrm{ch}}
\newcommand{\dep}{\mathrm{dep}}
\newcommand{\cl}{C}
\newcommand{\Star}{\bigstar}
\newaliascnt{openproblem}{theorem}
\newtheorem{openproblem}[openproblem]{Open Problem}
\aliascntresetthe{openproblem}

\title{nni+: a local subtree-move distance for mutation trees}
\titlerunning{nni+: a local subtree-move distance for mutation trees}

\author{Lu\'is Cunha}{Instituto de Computa\c{c}\~ao, Universidade Federal Fluminense, Brazil}{lfignacio@ic.uff.br}{https://orcid.org/0000-0002-3797-6053}{}
\author{Jack Kuipers}{D-BSSE, ETH Zurich, Basel, Switzerland}{jack.kuipers@bsse.ethz.ch}{https://orcid.org/0000-0001-5357-2705}{}
\author{Thiago Nascimento}{Instituto de Computa\c{c}\~ao, Universidade Federal Fluminense, Brazil}{thiago\_nascimento@id.uff.br}{https://orcid.org/0009-0005-5278-6679}{}

\authorrunning{L. Cunha, J. Kuipers, and T. Nascimento}
\Copyright{Lu\'is Cunha, Jack Kuipers, and Thiago Nascimento}
\ccsdesc[500]{Applied computing~Computational biology}
\ccsdesc[300]{Mathematics of computing~Trees}
\keywords{mutation trees, tumor phylogenetics, tree rearrangement distances, nearest neighbor interchange, exact algorithms}
\relatedversion{}

\EventEditors{}
\EventNoEds{0}
\EventLongTitle{}
\EventShortTitle{}
\EventAcronym{}
\EventYear{}
\EventDate{}
\EventLocation{}
\EventLogo{}
\SeriesVolume{}
\ArticleNo{}

\begin{document}
\nolinenumbers
\maketitle

\begin{abstract}
Mutation trees describe the order in which somatic mutations accumulate in a tumor. Inference from noisy sequencing data yields many plausible trees, so comparing mutation trees is a basic task, and a distance is most useful when its value has a direct evolutionary meaning. Rearrangement distances have this meaning: they count how many local corrections separate two evolutionary histories. However, the classical local move, the nearest neighbor interchange (NNI), does not suit mutation trees, where every node carries a mutation, because it preserves the number of children of every node and so cannot connect all trees. We propose \nnip, which moves a set of sibling subtrees one level up or down and corresponds to stating that a mutation occurred before or after a group of subclones branched off. We compare it with its elementary version, which moves one subtree at a time, and with the rotation-free generalized NNI (gNNI), which exchanges subtrees across an edge in both directions.

We show that \nnip induces a metric with diameter $2n-2$, and we give a 2-approximation for \nnip and for gNNI. Our main result is that the \nnip distance has an exact formula and can be computed in polynomial time, whereas both the elementary distance and the rotation-free gNNI distance are \NP-complete. The key is to view a mutation tree as a family of compatible clusters, in which every move replaces a single cluster. For the elementary distance we also characterize, by a linear-time test, when it coincides with the parent--child distance.
\end{abstract}

\section{Introduction}\label{sec:intro}

As a tumor grows, its cells accumulate somatic mutations, and the history of these events is commonly summarized by a \emph{mutation tree}: a rooted tree whose root is the normal (germline) state and whose every other node is labelled by a mutation, so that the genotype of a clone is the set of labels on its root path~\cite{kuipers2017advances,jahn2016tree}. Under the infinite-sites assumption every mutation is gained once and never lost, and mutation trees are exactly the perfect phylogenies of the underlying binary cell-by-mutation matrix, whose entry is $1$ when the cell carries the mutation~\cite{gusfield1991efficient} (\autoref{fig:tumor}). Unlike classical leaf-labelled phylogenies, \emph{internal nodes carry labels}: a mutation that is an ancestor of other mutations is observed, not inferred.

Single-cell data are noisy (false negatives from allelic dropout, false positives, missing entries), and probabilistic methods such as SCITE~\cite{jahn2016tree} therefore return a posterior sample of mutation trees rather than a single tree. The same need for comparison arises across patients, where recurrent evolutionary patterns are sought~\cite{morita2020clonal,christensen2020recap}. This motivates distances between mutation trees and consensus trees under those distances~\cite{govek2018consensus,aguse2019summarizing,christensen2020recap}.

Several distances for mutation trees exist: the parent--child and ancestor--descendant distances~\cite{govek2018consensus}, CASet and DISC~\cite{dinardo2020distance}, MLTD~\cite{karpov2019mltd}, the triplet-based MP3 similarity~\cite{ciccolella2021triplet}, the Bourque distances~\cite{jahn2021bourque}, and rearrangement distances on fully labelled trees~\cite{bernardini2019rearrangement}. Most of them count discrepancies between derived set systems. Rearrangement distances instead count \emph{edit operations}, which gives an interpretable answer to the question ``how many local changes separate two evolutionary scenarios?''. The answer is expressed in units of evolutionary statements. For the operation studied here, for instance, if an inferred tree is at distance two from the true tree, then two moves turn one into the other. Each move states that one mutation occurred earlier or later relative to some of its subclones. A count of discrepant ancestor pairs or triplets does not convey this, because a single misplaced mutation can change many pairs at once. A move also defines neighborhoods of trees, which is how tree spaces are explored in practice: SCITE samples mutation trees by MCMC with prune-and-reattach and node-swap moves~\cite{jahn2016tree}.

\subparagraph{Local moves on mutation trees.}
For leaf-labelled phylogenies the standard local move is the nearest neighbor interchange (NNI), whose distance is NP-hard to compute~\cite{dasgupta2000computing}. Transplanted to mutation trees, an NNI on an edge swaps a child of the lower endpoint with a sibling of that endpoint (\autoref{sec:ops}). Such a swap preserves the number of children of every mutation, so trees with different branching patterns are never connected by NNIs. Jahn, Beerenwinkel and Zhang~\cite{jahn2021bourque} introduced the generalized NNI (gNNI), which exchanges arbitrary subsets of subtrees across an edge and, on unrooted trees, is sandwiched by the Robinson--Foulds distance; on rooted trees they also allow a rotation of the edge. We study gNNI without the rotation (\autoref{sec:ops}). The prune-and-reattach move of SCITE~\cite{jahn2016tree} and the link-and-cut operation of~\cite{bernardini2019rearrangement} reattach a subtree anywhere, which is not local.

This paper studies \nnip, a local operation that connects the space of mutation trees: it moves a set of sibling subtrees by one level, either up to their grandparent or down into one of their siblings (\autoref{fig:moves}). Moving a single subtree is the \emph{elementary} case \enni. Every NNI is the composition of two elementary moves, every \nnip move is a gNNI move in which one side is empty, and every elementary move is a prune-and-reattach whose new attachment point is adjacent to the old one. The \nnip distance counts the minimum number of single-level relocations of groups of clades needed to turn one evolutionary scenario into another.

\subparagraph{Why study \nnip.}
Two requirements guide the choice of operation. First, one move should correspond to one plausible correction of an inferred tree. A typical inference error places a single mutation too early or too late with respect to some of its subclones. This displaces a group of sibling clades by one level at once, and \nnip undoes it with one move, whereas the elementary operation needs one move per clade and a count of discrepant pairs grows with the size of the clades. Second, the distance should be computable exactly, so that it can be used to compare and summarize many trees. The three operations studied here, \enni, \nnip and gNNI, form a hierarchy of increasing power. Determining which of them admit efficient exact algorithms shows what makes a rearrangement distance on mutation trees tractable, and it singles out \nnip as the operation that meets both requirements.

\subparagraph{Contributions.}
\begin{itemize}
\item We define \nnip and its elementary version on mutation trees and relate them formally to NNI, gNNI, link-and-cut and the parent--child distance (\autoref{sec:ops}).
\item We prove that \nnip induces a metric, compute the distance to the star exactly, show that the diameter is exactly $2n-2$, give simple lower bounds, and prove that contracting the mutations whose clusters differ is a 2-approximation, for \nnip and for gNNI (\autoref{sec:basic}).
\item For the elementary operation, we show that its distance equals the parent--child (and link-and-cut) distance exactly when no mutation is dead, a local condition testable in linear time, and we give lower bounds on the excess (\autoref{sec:tight}). We prove that the elementary distance is \NP-complete, by a reduction from \textsc{Vertex Cover} built from precedence gadgets whose order constraints form one cycle per edge (\autoref{sec:complexity}).
\item We show that mutation trees are exactly the families of pairwise compatible clusters, and that \nnip moves and (rotation-free) gNNI moves are exactly the replacements of a single cluster that only remove or only add labels, and arbitrary replacements of a single cluster, respectively. This yields exact formulas. For \nnip, $\dn=\mathrm N+\nu$, where $\mathrm N$ counts the required cluster losses and gains and $\nu$ is a maximum matching in a bipartite graph of tangles, computable in $O(n^3)$ time. For gNNI, $\dgnni=\dcl+\mathrm{fvs}$, where $\dcl$ is the number of mutations whose cluster differs in the two trees and $\mathrm{fvs}$ is a minimum feedback vertex set of a conflict digraph. From the latter we derive that the rotation-free gNNI distance is \NP-complete, \textsf{APX}-hard and fixed-parameter tractable in the excess over $\dcl$ (\autoref{sec:setcomplexity}). The variant with rotations, also defined in~\cite{jahn2021bourque}, is not covered by these results.
\end{itemize}

\section{Mutation trees and the nni+ operation}\label{sec:ops}

\subparagraph{Mutation trees.}
Let $L=\{1,\ldots,n\}$ be a set of mutations. A \emph{mutation tree} on $L$ is a rooted tree $T$ with vertex set $L\cup\{0\}$ and root $0$, the normal cell. For $v\in L$ we write $\pa_T(v)$ for its parent, $\ch_T(u)$ for the set of children of $u$, $\dep_T(v)$ for the number of edges between $0$ and $v$, and $\cl_T(v)\subseteq L$ for the \emph{cluster} of $v$, the set of labels in the subtree rooted at $v$ (including $v$). This subtree, that is, $v$ together with all its descendants, is the \emph{clade} of $v$; the cluster is its set of labels. We denote by $r(T)=|\ch_T(0)|$ the number of children of the root, and by $\Star$ the \emph{star}, the tree in which every mutation is a child of the root. $\T_n$ is the set of mutation trees on $L$; by Cayley's formula $|\T_n|=(n+1)^{n-1}$. Throughout, both trees of a comparison have the same label set; extending the framework to trees whose nodes carry sets of mutations, or to different label sets~\cite{jahn2021bourque,karpov2019mltd}, is discussed in \autoref{sec:conclusion}.

\autoref{fig:tumor} shows a mutation tree together with an error-free and a noisy cell-by-mutation matrix. The error-free matrix admits a perfect phylogeny: the sets of cells carrying each mutation are pairwise nested or disjoint. The tree is obtained directly from the matrix, since the parent of each mutation is the mutation whose set of cells is the smallest one strictly containing its own (or the root, if there is none).

\begin{figure}[t]
\centering
\begin{tikzpicture}[font=\small]
\node at (0,2.2) {\textbf{(i)} observed};
\node at (0,0.4) {\begin{tabular}{c|ccccc}
 & 1 & 2 & 3 & 4 & 5\\\hline
$c_1$ & 1 & 0 & 0 & 0 & 0\\
$c_2$ & 1 & 1 & 0 & \textcolor{red}{\textbf{0}} & 0\\
$c_3$ & \textcolor{red}{\textbf{?}} & 1 & 0 & 0 & 0\\
$c_4$ & 1 & 0 & 1 & 0 & 1\\
$c_5$ & 1 & \textcolor{red}{\textbf{1}} & 1 & 0 & 0\\
$c_6$ & 0 & 0 & 0 & 0 & 0
\end{tabular}};
\node at (4.0,2.2) {\textbf{(ii)} error free};
\node at (4.0,0.4) {\begin{tabular}{c|ccccc}
 & 1 & 2 & 3 & 4 & 5\\\hline
$c_1$ & 1 & 0 & 0 & 0 & 0\\
$c_2$ & 1 & 1 & 0 & 1 & 0\\
$c_3$ & 1 & 1 & 0 & 0 & 0\\
$c_4$ & 1 & 0 & 1 & 0 & 1\\
$c_5$ & 1 & 0 & 1 & 0 & 0\\
$c_6$ & 0 & 0 & 0 & 0 & 0
\end{tabular}};
\begin{scope}[xshift=8.3cm,yshift=1.9cm,every node/.style={circle,draw,inner sep=1.5pt,minimum size=13pt}]
\node at (-0.5,0.3) [draw=none,rectangle] {\textbf{(iii)}};
\node (r) at (0.6,0) {$0$};
\node (a) at (0.6,-0.9) {$1$};
\node (b) at (-0.1,-1.8) {$2$};
\node (c) at (1.3,-1.8) {$3$};
\node (d) at (-0.1,-2.7) {$4$};
\node (e) at (1.3,-2.7) {$5$};
\draw (r)--(a) (a)--(b) (a)--(c) (b)--(d) (c)--(e);
\begin{scope}[every node/.style={rectangle,draw,fill=gray!20,inner sep=1pt,font=\scriptsize}]
\node (k6) at (1.6,0) {$c_6$}; \draw[dashed] (r)--(k6);
\node (k1) at (1.6,-0.9) {$c_1$}; \draw[dashed] (a)--(k1);
\node (k3) at (-0.9,-1.8) {$c_3$}; \draw[dashed] (b)--(k3);
\node (k5) at (2.1,-1.8) {$c_5$}; \draw[dashed] (c)--(k5);
\node (k2) at (-0.9,-2.7) {$c_2$}; \draw[dashed] (d)--(k2);
\node (k4) at (2.1,-2.7) {$c_4$}; \draw[dashed] (e)--(k4);
\end{scope}
\end{scope}
\end{tikzpicture}
\caption{\textbf{(i)} Observed genotypes of six cells (rows) at five mutated sites (columns). Red entries are errors with respect to (ii): a false negative ($c_2$, site 4), a false positive ($c_5$, site 2) and a missing value ($c_3$, site 1). \textbf{(ii)} The error-free matrix; for every two mutations the sets of cells carrying them are nested or disjoint, so the matrix has a perfect phylogeny. \textbf{(iii)} The corresponding mutation tree: node $0$ is the normal state, every other node is a mutation, and each cell (grey) attaches to the node whose root path is its genotype. Inference methods such as SCITE work on trees like (iii), with labels on all nodes.}\label{fig:tumor}
\end{figure}

\subparagraph{The nni+ move.}
\begin{definition}\label{def:nniplus}
Let $T\in\T_n$, let $u\in L\cup\{0\}$, and let $J\subseteq\ch_T(u)$ be nonempty. An \emph{\nnip move} on $u$ with set $J$ is one of the following:
\begin{itemize}
\item an \emph{up-move}: provided $u\neq0$, every $v\in J$ becomes a child of $\pa_T(u)$;
\item a \emph{down-move into $w$}: for some $w\in\ch_T(u)\setminus J$, every $v\in J$ becomes a child of $w$.
\end{itemize}
For $S,T\in\T_n$, the \emph{\nnip distance} $\dn(S,T)$ is the minimum number of \nnip moves transforming $S$ into $T$. A move with $|J|=1$ is called \emph{elementary}, and the distance that only uses elementary moves is denoted $\de(S,T)$; we write \enni for this operation.
\end{definition}

In words, an \nnip move takes several sibling clades and moves them together one level up, to their grandparent, or one level down, into one of their siblings (\autoref{fig:moves}). Subtrees travel with their roots, no node is created or deleted, and the new parent is never a descendant of a moved node, so the result is again a mutation tree. Biologically, a down-move of $J$ into $w$ states that mutation $w$ occurred before the clones $J$ branched off, and an up-move states the converse. Several subclones are thus re-attached in one event, which is what a single misplaced mutation in an inferred tree typically causes.

\begin{figure}[t]
\centering
\begin{tikzpicture}[font=\small,every node/.style={circle,draw,inner sep=1pt,minimum size=12pt},
  tri/.style={draw,regular polygon,regular polygon sides=3,minimum size=14pt,inner sep=0pt}]
\begin{scope}
\node[draw=none] at (0.3,0.7) {\textbf{(i)} $T$};
\node (g) at (0.3,0) {$g$}; \node (u) at (0.3,-0.9) {$u$};
\node[fill=red!25] (v1) at (-0.6,-1.8) {$v_1$}; \node[fill=red!25] (v2) at (0.1,-1.8) {$v_2$}; \node (x) at (0.8,-1.8) {$x$}; \node (w) at (1.5,-1.8) {$w$};
\node[tri] (A) at (-0.6,-2.6) {}; \node[tri] (B) at (0.1,-2.6) {};
\draw (g)--(u) (u)--(v1) (u)--(v2) (u)--(x) (u)--(w) (v1)--(A) (v2)--(B);
\end{scope}
\begin{scope}[xshift=4.6cm]
\node[draw=none] at (0.3,0.7) {\textbf{(ii)} up-move of $\{v_1,v_2\}$};
\node (g) at (0.3,0) {$g$}; \node (u) at (1.1,-0.9) {$u$};
\node[fill=red!25] (v1) at (-0.6,-0.9) {$v_1$}; \node[fill=red!25] (v2) at (0.25,-0.9) {$v_2$};
\node (x) at (0.8,-1.8) {$x$}; \node (w) at (1.5,-1.8) {$w$};
\node[tri] (A) at (-0.6,-1.7) {}; \node[tri] (B) at (0.25,-1.7) {};
\draw (g)--(u) (g)--(v1) (g)--(v2) (u)--(x) (u)--(w) (v1)--(A) (v2)--(B);
\end{scope}
\begin{scope}[xshift=9.4cm]
\node[draw=none] at (0.3,0.7) {\textbf{(iii)} down-move of $\{v_1,v_2\}$ into $w$};
\node (g) at (0.3,0) {$g$}; \node (u) at (0.3,-0.9) {$u$};
\node (x) at (-0.3,-1.8) {$x$}; \node (w) at (0.9,-1.8) {$w$};
\node[fill=red!25] (v1) at (0.5,-2.7) {$v_1$}; \node[fill=red!25] (v2) at (1.3,-2.7) {$v_2$};
\node[tri] (A) at (0.5,-3.5) {}; \node[tri] (B) at (1.3,-3.5) {};
\draw (g)--(u) (u)--(x) (u)--(w) (w)--(v1) (w)--(v2) (v1)--(A) (v2)--(B);
\end{scope}
\end{tikzpicture}
\caption{The two kinds of \nnip move. Trees (ii) and (iii) are each obtained from (i) by one move on $u$ with $J=\{v_1,v_2\}$. Triangles are subtrees, which travel with their roots. (ii) is undone by moving $J$ down into $u$, and (iii) is undone by moving $J$ up from $w$. Both moves relocate two clades at once, so they are not elementary; an elementary move is the case $|J|=1$.}\label{fig:moves}
\end{figure}

\subparagraph{Related operations.}
Every operation below defines a \emph{move graph} on $\T_n$, in which two trees are adjacent when one move turns one into the other, and a distance, the length of a shortest path in that graph ($\infty$ if there is none). We compare \nnip with the following operations and distances.
\begin{itemize}
\item \emph{NNI.} For an edge $(u,w)$ with $w\in\ch_T(u)$, exchange a child $x$ of $w$ with a child $y\neq w$ of $u$. This is the rooted analogue of the classical NNI~\cite{dasgupta2000computing,li1996nni}.
\item \emph{Elementary \nnip} (\enni), the case $|J|=1$ of \autoref{def:nniplus}.
\item \emph{gNNI.} For an edge $(u,w)$ with $w\in\ch_T(u)$, move a set $J_u\subseteq\ch_T(u)\setminus\{w\}$ into $w$ and a set $J_w\subseteq\ch_T(w)$ up to $u$, simultaneously and not both empty. This is the generalized NNI of~\cite{jahn2021bourque} on rooted trees \emph{without rotation}. Jahn et al.\ also allow the move to exchange the positions of $u$ and $w$ (so that $w$ takes the place of $u$ and $u$ becomes a child of $w$) while redistributing the children. We write $d_{\mathrm{gNNI}}^{\mathrm{rot}}$ for the distance with rotations. Throughout, gNNI means the rotation-free operation. It keeps the ancestral order of $u$ and $w$ and re-specifies only the content of the clade of $w$, which is exactly a change of one cluster (\autoref{lem:onecluster}). Rotations are discussed briefly in \autoref{sec:setcomplexity}.
\item \emph{Contraction and expansion.} A \emph{contraction} of $u\neq0$ is the up-move on $u$ with $J=\ch(u)$, which makes $u$ a leaf. An \emph{expansion} of a leaf $w$ is the inverse: a down-move into $w$, where $w$ has no children. These are special \nnip moves and are the analogues of the edge contractions and expansions that realize the Robinson--Foulds distance on leaf-labelled trees~\cite{robinson1981comparison}. We denote their distance by $d_{\mathrm{CE}}$.
\item \emph{link-and-cut}~\cite{bernardini2019rearrangement}. Make $v$ a child of any node that is not its descendant. This is the prune-and-reattach move of SCITE~\cite{jahn2016tree}.
\item The \emph{parent--child distance} $\dpc(S,T)=|\{v\in L:\pa_S(v)\neq\pa_T(v)\}|$~\cite{govek2018consensus}.
\end{itemize}

\begin{proposition}\label{prop:relations}
For all $S,T\in\T_n$:
\begin{enumerate}
\item Every NNI move preserves $|\ch(x)|$ for every node $x$. Hence the NNI move graph on $\T_n$ is disconnected for $n\ge2$.
\item Every NNI move is the composition of two elementary moves. Hence $\de(S,T)\le 2\,\dnni(S,T)$ whenever $\dnni(S,T)$ is finite.
\item $d_{\mathrm{gNNI}}^{\mathrm{rot}}(S,T)\le\dgnni(S,T)\le\dn(S,T)\le\min\{2\,\dgnni(S,T),\ \de(S,T),\ d_{\mathrm{CE}}(S,T)\}$.
\item The link-and-cut distance equals $\dpc(S,T)$, and $\dpc(S,T)\le\de(S,T)$.
\end{enumerate}
\end{proposition}
\begin{proof}
(1) An NNI on the edge $(u,w)$ exchanges a child $x$ of $w$ with a child $y\neq w$ of $u$: $w$ loses $x$ and gains $y$, $u$ loses $y$ and gains $x$, and all other child sets are untouched. So every node keeps its number of children. The star has $|\ch(0)|=n$, while the chain (path graph) $0\to1\to\cdots\to n$ has $|\ch(0)|=1$, so for $n\ge2$ no sequence of NNI moves connects them. An NNI always exchanges two subtrees, one on each side of the edge, as in the classical definition~\cite{li1996nni,jahn2021bourque}. If one side were allowed to be empty, the move would relocate a single subtree by one level. That is exactly an elementary \nnip move, and it does connect the space.

(2) The exchange of $x$ and $y$ in (1) is realized by two elementary moves: first move $x$ up from $w$ to its grandparent $u$, then move $y$ down into its sibling $w$. Replacing each NNI move of a shortest NNI sequence by these two moves gives the bound.

(3) A gNNI move with rotation disabled is a special case of the move of~\cite{jahn2021bourque}, which gives the first inequality. An elementary move and a contraction or expansion are \nnip moves. An \nnip move is a gNNI move with $J_w=\emptyset$ (down-move into $w$) or with $J_u=\emptyset$ (up-move from $w$ to $u$). Conversely, a gNNI move on $(u,w)$ is the up-move of $J_w$ from $w$ to $u$ followed by the down-move of $J_u$ (the original set) into $w$.

(4) Each elementary move changes the parent of exactly one mutation, so $\de\ge\dpc$. Bernardini et al.~\cite[Lemma 10]{bernardini2019rearrangement} show that the link-and-cut distance equals the number of mutations with different parents. (Directly: process the mutations of $T$ in breadth-first order and give each one its parent in $T$. When $v$ is processed, its target parent is already settled on a path to the root, which by construction does not pass through $v$.)
\end{proof}

The gap between $\dn$ and $\de$ is unbounded. If $S$ has $k$ leaves below $a$ that are children of its sibling $b$ in $T$, then $\de(S,T)\ge\dpc(S,T)=k$, while $\dn(S,T)=2$: move all $k$ leaves up from $a$, then down into $b$. This is why \nnip, and not \enni, is the right local operation for comparing clone trees: relocating a clade together with its sibling subclones is one evolutionary statement, not $k$ of them.

\section{Basic properties, lower bounds and a 2-approximation}\label{sec:basic}

Consider an \nnip move on $u$ with set $J$, and let $C=\bigcup_{v\in J}\cl_T(v)$ be the labels of the moved clades. Only the parent pointers of the nodes of $J$ are reassigned, and the moved clades travel intact, so the nodes of $C$ shift one level up or down and nothing else moves. For the clusters, consider an up-move from $u$ to $g=\pa(u)$. A cluster changes only if it contains labels of $C$ before or after the move. Before the move these are the clusters of $u$ and its ancestors, and after it the clusters of the ancestors of $u$, which still contain $C$. Hence only $\cl(u)$ changes, and it loses exactly $C$. A down-move into $w$ is the inverse of the up-move of $J$ from $w$ (\autoref{thm:metric}(1)), so only $\cl(w)$ changes, and it gains $C$. We record these facts.

\begin{observation}\label{lem:effect}
Let $T'$ be obtained from $T$ by an \nnip move on $u$ with set $J$, and let $C=\bigcup_{v\in J}\cl_T(v)$. Then:
\begin{enumerate}
\item the mutations whose parent changes are exactly those of $J$;
\item exactly one cluster changes: $\cl(u)$ loses $C$ for an up-move, and $\cl(w)$ gains $C$ for a down-move into $w$;
\item the depth of every node of $C$ decreases by one (up-move) or increases by one (down-move), and all other depths are unchanged.
\end{enumerate}
\end{observation}

It helps to think of a tree as its \emph{ancestry matrix}, with one row and one column per mutation, and entry $(x,z)$ equal to $1$ when $x\in\cl(z)$. The column of $z$ lists the cluster of $z$, that is, $z$ and its descendants. If every node carries one cell, this is an error-free cell-by-mutation matrix like the one in \autoref{fig:tumor}(ii), and the column of $z$ lists the cells that carry mutation $z$. By \autoref{lem:effect}(2), a move changes exactly one column. An up-move on $u$ removes entries from the column of $u$, and we call it a \emph{loss} on column $u$. A down-move into $w$ adds entries to the column of $w$, and we call it a \emph{gain} on column $w$. We call either change an \emph{operation} on that column.

\begin{theorem}\label{thm:metric}
For all $S,T\in\T_n$:
\begin{enumerate}
\item every \nnip move can be undone by an \nnip move;
\item $\dn(T,\Star)=\iota(T)$, where $\iota(T)=|\{v\in L:\ch_T(v)\neq\emptyset\}|$ is the number of non-root internal nodes;
\item $\dn$ is a metric on $\T_n$, and $\dn(S,T)\le\iota(S)+\iota(T)\le 2n-2$.
\end{enumerate}
\end{theorem}
\begin{proof}
(1) After an up-move of $J$ from $u$ to $g=\pa(u)$, the node $u$ is a sibling of every $v\in J$, and moving $J$ down into $u$ restores $T$. After a down-move of $J$ into $w$, moving $J$ up from $w$ restores the parent $u=\pa(w)$.

(2) In $\Star$ every cluster is a singleton, so every non-root internal node $v$ of $T$ must lose its descendants, which costs at least one operation on column $v$; hence $\dn(T,\Star)\ge\iota(T)$. Conversely, process the internal nodes of $T$ by non-increasing depth and, for each, move all its current children up to its parent. When $v$ is processed its descendants have already been flattened into children of $v$, so one move makes $v$ a leaf, and this never gives children to a node already processed (only the parent of $v$, which is shallower, gains children). After $\iota(T)$ moves we reach $\Star$.

(3) By (1) and (2) the move graph is undirected and connected, and $\dn$ is its shortest-path metric. The bound follows from the triangle inequality through $\Star$ and $\iota(T)\le n-1$.
\end{proof}

\subparagraph{Lower bounds.}
For a mutation $z$ let $L_z=\cl_S(z)\setminus\cl_T(z)$ and $G_z=\cl_T(z)\setminus\cl_S(z)$ be what its cluster has to lose and to gain, and let $\mathrm{need}(z)=[L_z\neq\emptyset]+[G_z\neq\emptyset]\in\{0,1,2\}$, where $[P]$ is $1$ if the statement $P$ holds and $0$ otherwise. So $\mathrm{need}(z)$ counts whether column $z$ has something to lose and whether it has something to gain. Define
\[
\dcl(S,T)=\bigl|\{v\in L:\cl_S(v)\neq\cl_T(v)\}\bigr|,\qquad
\mathrm{N}(S,T)=\sum_{z\in L}\mathrm{need}(z),\]
\[
\ddep(S,T)=\max_{v\in L}\bigl(\dep_S(v)-\dep_T(v)\bigr)^+ +\max_{v\in L}\bigl(\dep_T(v)-\dep_S(v)\bigr)^+,
\]
where $x^+=\max(x,0)$. The first maximum is the largest distance some mutation has to rise, and the second the largest distance some mutation has to sink. We add the two because an up-move only decreases depths and a down-move only increases them, so no move can reduce both at once (for a single maximum of $|\dep_S(v)-\dep_T(v)|$ the bound would be weaker). We call $\dcl$ the \emph{cluster distance}: it counts the mutations whose cluster differs. It compares the clusters of the \emph{same label} in the two trees; a cluster of $S$ that is missing from $T$ belongs to a label whose cluster changed, so $\dcl$ is at least half the cluster-based Robinson--Foulds distance~\cite{robinson1981comparison}. Clearly $\dcl\le\mathrm N\le2\dcl$.

\begin{theorem}\label{thm:lb}
Let $\LB(S,T)=\max\{\mathrm N(S,T),\ddep(S,T)\}$. Then $\dn(S,T)\ge\LB(S,T)\ge\dcl(S,T)$. Moreover, if $S'$ is obtained from $S$ by one move, then $\LB(S,T)\le 1+\LB(S',T)$.
\end{theorem}
\begin{proof}
A move is one operation on one column $z$. A loss only removes elements from $\cl(z)$ and a gain only adds elements, so a single operation can make $L_z$ empty or $G_z$ empty, but not both, and it leaves every other column unchanged. Hence $\mathrm N$ decreases by at most one per move, and it is $0$ only at $T$. For the depth bound, by \autoref{lem:effect}(3) an up-move decreases depths by one and increases none, so the first maximum in $\ddep$ drops by at most one and the second does not drop; down-moves are symmetric. So each of the two bounds decreases by at most one per move, and so does their maximum.
\end{proof}

The second statement says that $\LB$ changes by at most one per move, so it is a consistent heuristic for an exact A* search of the elementary distance. $\mathrm N$ is a \emph{travel bound}: if a mutation $x$ has to leave the subtrees of $\ell$ ancestors, each of these $\ell$ columns needs its own loss. So the journey of $x$ costs at least $\ell$ moves, even if other mutations travel together with $x$ in the same moves. For the two chains $0\to1\to\cdots\to n$ and $0\to n\to\cdots\to1$ we have $\ddep=2n-2$, so with \autoref{thm:metric}(3):

\begin{theorem}\label{thm:diameter}
For $n\ge1$, $\max_{S,T\in\T_n}\dn(S,T)=2n-2$, attained by the chain $0\to1\to\cdots\to n$ and the reversed chain.
\end{theorem}

\subparagraph{A 2-approximation.}
The simplest lower bound, $\dcl$, is always within a factor two of the distance.

\begin{theorem}\label{thm:approx}
Let $X=\{v\in L:\cl_S(v)\neq\cl_T(v)\}$. Then
\[
\dcl(S,T)\ \le\ \dn(S,T)\ \le\ |X\cap I(S)|+|X\cap I(T)|\ \le\ 2\,\dcl(S,T),
\]
where $I(R)$ is the set of non-root internal nodes of $R$. A sequence of this length is computed in $O(n^2)$ time.
\end{theorem}
\begin{proof}
The first inequality is \autoref{thm:lb}. For the second, \emph{contract} the nodes of $X$ in $S$: process $X\cap I(S)$ by non-increasing depth in $S$ and move all current children of the processed node $v$ up to its parent. As in \autoref{thm:metric}(2), each contraction is one move, and a node of $X$ that is a leaf of $S$ never gains children (it has no descendants). By \autoref{lem:effect}(2), contracting $v$ changes only $\cl(v)$, which becomes $\{v\}$. Let $S'$ be the result: $\cl_{S'}(v)=\{v\}$ for $v\in X$ and $\cl_{S'}(v)=\cl_S(v)=\cl_T(v)$ for $v\notin X$. Contract $T$ in the same way to obtain $T'$. Then $S'$ and $T'$ have the same cluster for every label. Since $x$ is an ancestor of $y$ exactly when $y\in\cl(x)$, they are the same tree. Hence $\dn(S,T)\le\dn(S,S')+\dn(T',T)\le|X\cap I(S)|+|X\cap I(T)|$, using \autoref{thm:metric}(1) to reverse the contraction of $T$. Clusters are bit vectors, so each step costs $O(n)$.
\end{proof}

The contraction sequence of the proof uses only contractions and expansions, and for these moves the bound is exact. This is the mutation-tree analogue of the Robinson--Foulds distance, which on leaf-labelled trees counts the edge contractions and expansions needed.

\begin{proposition}\label{prop:ce}
$d_{\mathrm{CE}}(S,T)=|X\cap I(S)|+|X\cap I(T)|$.
\end{proposition}
\begin{proof}
The upper bound is the sequence in the proof of \autoref{thm:approx}. For the lower bound, each contraction or expansion changes exactly one cluster (\autoref{lem:effect}). A contraction turns a cluster into a singleton, and an expansion turns a singleton into a larger set. Let $v\in X\cap I(S)$. Its cluster $\cl_S(v)$ is not a singleton and has to change, and its first change can only be a contraction of $v$. Let $v\in X\cap I(T)$. Its final cluster $\cl_T(v)$ is not a singleton and differs from the initial one, so its last change can only be an expansion of $v$. A label in both sets needs two different moves.
\end{proof}

So $\dn\le d_{\mathrm{CE}}\le2\dcl$. An \nnip move generalizes a contraction by moving only \emph{some} of the children of $u$, and an expansion by allowing the target $w$ to have children. In \autoref{fig:approx} the two coincide. They differ as soon as a clade keeps part of its children. Suppose that $a$ has the leaf children $x,y$, its sibling $b$ has the leaf child $z$, and $T$ is obtained by moving $x$ from $a$ to $b$. Then $\dn=2$ (move $x$ up, then into $b$), while $d_{\mathrm{CE}}=4$, because $a$ and $b$ must each be contracted and expanded again.

\begin{figure}[t]
\centering
\begin{tikzpicture}[font=\scriptsize,every node/.style={circle,draw,inner sep=0.5pt,minimum size=11pt},
  xx/.style={fill=red!20}]
\begin{scope}
\node[draw=none] at (0.6,0.55) {\textbf{(i)} $S$};
\node (r) at (0.6,0) {$0$}; \node[xx] (a) at (0,-0.8) {$a$}; \node (b) at (1.2,-0.8) {$b$};
\node (c) at (-0.4,-1.6) {$c$}; \node (d) at (0.4,-1.6) {$d$};
\draw (r)--(a) (r)--(b) (a)--(c) (a)--(d);
\end{scope}
\begin{scope}[xshift=3cm]
\node[draw=none] at (0.9,0.55) {\textbf{(ii)} $S'=T'$};
\node (r) at (0.9,0) {$0$}; \node[xx] (a) at (0,-0.8) {$a$}; \node (c) at (0.6,-0.8) {$c$}; \node (d) at (1.2,-0.8) {$d$}; \node[xx] (b) at (1.8,-0.8) {$b$};
\draw (r)--(a) (r)--(b) (r)--(c) (r)--(d);
\end{scope}
\begin{scope}[xshift=6.6cm]
\node[draw=none] at (0.6,0.55) {\textbf{(iii)} $T$};
\node (r) at (0.6,0) {$0$}; \node (a) at (0,-0.8) {$a$}; \node[xx] (b) at (1.2,-0.8) {$b$};
\node (c) at (0.8,-1.6) {$c$}; \node (d) at (1.6,-1.6) {$d$};
\draw (r)--(a) (r)--(b) (b)--(c) (b)--(d);
\end{scope}
\draw[->,thick,draw=gray] (1.7,-0.9)--(2.7,-0.9); \draw[->,thick,draw=gray] (5.6,-0.9)--(6.5,-0.9);
\node[draw=none] at (2.2,-1.2) {1 move}; \node[draw=none] at (6.05,-1.2) {1 move};
\end{tikzpicture}
\caption{The 2-approximation of \autoref{thm:approx}. The labels whose clusters differ ($X=\{a,b\}$, red) are contracted in both trees, which makes the trees equal. Here the bound gives $\dn\le2$, which is exact: $\{c,d\}$ moves up from $a$ and then down into $b$. In contrast, $\de=4$.}\label{fig:approx}
\end{figure}

Since $\dgnni\le\dn\le2\dgnni$ (\autoref{prop:relations}) and every gNNI move also changes exactly one cluster, the same contraction gives $\dcl\le\dgnni\le2\dcl$, a 2-approximation for gNNI as well. This is the analogue, with the label-wise count $\dcl$, of the bounds $\tfrac12\mathrm{RF}\le d\le\mathrm{RF}$ that Jahn et al.~\cite{jahn2021bourque} prove for gNNI with rotations, where RF is the size of the symmetric difference of the two cluster families. For the elementary operation no such bound exists: the example after \autoref{prop:relations} has $\dcl=2$ and $\de\ge k$. The 2-approximation is not tight in general. \autoref{sec:setcomplexity} gives exact formulas for both distances: polynomial for \nnip, and \NP-hard but fixed-parameter tractable for the rotation-free gNNI.

\section{The elementary operation: when is one move per misplaced mutation enough?}\label{sec:tight}

In this section and the next we study the elementary operation \enni, which relocates one clade at a time. It is the most local case of the prune-and-reattach move of the MCMC sampler SCITE~\cite{jahn2016tree}, which reattaches a clade anywhere. Its distance is an upper bound on $\dn$ (\autoref{prop:relations}), and its structure shows where the difficulty of local rearrangement distances comes from. As in \autoref{thm:metric}, moving every mutation of depth at least two up, one at a time and shallowest first, gives $\de(T,\Star)=n-r(T)$, where $r(T)$ is the number of children of the root, so $\de(S,T)\le2n-2$.

Write $s(v)=\pa_S(v)$, $t(v)=\pa_T(v)$, and let $M=\{v\in L:s(v)\neq t(v)\}$ be the set of \emph{misplaced} mutations, so $|M|=\dpc(S,T)$. Every misplaced mutation has to move at least once, so $\de\ge\dpc$. We characterize the pairs with $\de(S,T)=\dpc(S,T)$, which are exactly the pairs for which the local operation is as efficient as unrestricted link-and-cut (\autoref{prop:relations}(3)).

\subparagraph{The idea.}
In a sequence of exactly $|M|$ moves, every misplaced mutation $v$ moves once, directly from $s(v)$ to $t(v)$, and no other mutation moves. We call this move the \emph{jump} of $v$. A jump is an elementary move only if, at that moment, the old and the new parent are adjacent. Either $t(v)$ is the parent of $s(v)$, and $v$ moves up to its grandparent (\autoref{fig:jumps}(a)), or $s(v)$ is the parent of $t(v)$, and $v$ moves down into its sibling $t(v)$ (\autoref{fig:jumps}(b)). Whether this adjacency holds depends on the current parent of one other mutation, which we call the \emph{witness} of the jump: $s(v)$ for an up-jump (its parent must be $t(v)$) and $t(v)$ for a down-jump (its parent must be $s(v)$). If the witness is itself misplaced, its parent changes when it makes its own jump. So every misplaced mutation $v$ falls into one of four cases.
\begin{itemize}
\item $v$ is \emph{free}: the adjacency holds in $S$ through a mutation that is not misplaced, so it holds at all times and $v$ can jump whenever we like.
\item $v$ must jump \emph{before} a misplaced neighbour: the adjacency holds only while that neighbour is still in its initial position (\autoref{fig:jumps}(c)).
\item $v$ must jump \emph{after} a misplaced neighbour: the adjacency holds only once that neighbour is in its final position.
\item $v$ is \emph{dead}: the adjacency holds neither before nor after the neighbour jumps, so $v$ cannot reach $t(v)$ in one move and needs a detour (\autoref{fig:jumps}(d)).
\end{itemize}
The middle two cases are order requirements between two jumps. Such requirements can be impossible to meet when they form a cycle, for example $v$ before $u$, $u$ before $w$, and $w$ before $v$. The main observation of this section is that no cycle can occur here. Each requirement points from a node to one that is at most as deep in $S$, or at most as deep in $T$, and this rules out cycles (\autoref{lem:acyclic}). Consequently, the only obstruction is local (\autoref{thm:tight}): $\de=\dpc$ if and only if no mutation is dead. Cycles of order requirements do arise once mutations take detours, and \autoref{sec:complexity} shows that resolving them is \NP-hard.

\begin{figure}[t]
\centering
\resizebox{\linewidth}{!}{\begin{tikzpicture}[font=\scriptsize,
  m/.style={circle,draw,inner sep=0.5pt,minimum size=11pt},
  r/.style={m,fill=red!25},
  j/.style={->,thick,red!70!black,densely dashed},
  lbl/.style={font=\scriptsize\bfseries}]
\begin{scope}
\node[lbl] at (0.4,0.55) {(a) up-jump};
\node[m] (g) at (0,0) {$g$}; \node[m] (u) at (0,-0.8) {$u$}; \node[r] (v) at (0,-1.6) {$v$};
\draw (g)--(u)--(v);
\draw[j] (v) to[bend right=45] (g);
\node[align=left,anchor=north west] at (-0.9,-2.0) {$s(v)=u$, $t(v)=g$:\\ needs $\pa(u)=g$\\ at that moment};
\end{scope}
\begin{scope}[xshift=3.3cm]
\node[lbl] at (0.4,0.55) {(b) down-jump};
\node[m] (u) at (0,0) {$u$}; \node[r] (v) at (-0.5,-0.8) {$v$}; \node[m] (w) at (0.5,-0.8) {$w$};
\draw (u)--(v) (u)--(w);
\draw[j] (v) to[bend right=40] (w);
\node[align=left,anchor=north west] at (-0.9,-2.0) {$s(v)=u$, $t(v)=w$:\\ needs $\pa(w)=u$\\ at that moment};
\end{scope}
\begin{scope}[xshift=6.6cm]
\node[lbl] at (0.9,0.55) {(c) an order constraint};
\node[m] (g) at (0,0) {$g$}; \node[m] (u) at (0,-0.8) {$u$}; \node[r] (v) at (0,-1.6) {$v$};
\draw (g)--(u)--(v);
\node at (0.75,-0.8) {$\Rightarrow$};
\node[m] (g2) at (1.5,0) {$g$}; \node[m] (v2) at (1.5,-0.8) {$v$}; \node[m] (u2) at (1.5,-1.6) {$u$};
\draw (g2)--(v2)--(u2);
\node[align=left,anchor=north west] at (-0.6,-2.0) {$v$ jumps up while $u$ is\\ still below $g$; then $u$\\ jumps down into $v$.\\ $v$ must go first.};
\end{scope}
\begin{scope}[xshift=10.3cm]
\node[lbl] at (0.9,0.55) {(d) a dead mutation};
\node[m] (o) at (0.35,0) {$0$}; \node[m] (a) at (0,-0.8) {$a$}; \node[m] (b) at (0.7,-0.8) {$b$}; \node[r] (x) at (0,-1.6) {$x$};
\draw (o)--(a)--(x) (o)--(b);
\node at (1.15,-0.8) {$\Rightarrow$};
\node[m] (o2) at (1.95,0) {$0$}; \node[m] (a2) at (1.6,-0.8) {$a$}; \node[m] (b2) at (2.3,-0.8) {$b$}; \node[m] (x2) at (2.3,-1.6) {$x$};
\draw (o2)--(a2) (o2)--(b2)--(x2);
\node[align=left,anchor=north west] at (-0.4,-2.0) {$a,b$ never move and are\\ never adjacent: $x$ needs\\ two moves, $\de=2>\dpc=1$.};
\end{scope}
\end{tikzpicture}}
\caption{Single jumps. A misplaced mutation $v$ can reach $t(v)$ in one move only if its old and new parent are adjacent at that moment: (a) $t(v)$ is the parent of $s(v)$, or (b) $t(v)$ is a child of $s(v)$. (c) If the adjacency is provided by a neighbour that is itself misplaced, the two jumps must happen in a fixed order (here $v$ before $u$, recorded as $\beta(v)=u$ and as $\alpha(u)=v$). (d) If no adjacency is ever available, the mutation is dead.}\label{fig:jumps}
\end{figure}

\begin{lemma}\label{lem:schedule}
$\de(S,T)=|M|$ if and only if there is an ordering $v_1,\dots,v_m$ of $M$ such that, starting from $S$, for every $i$ the operation ``make $t(v_i)$ the parent of $v_i$'' is an elementary move in the current tree $T_{i-1}$.
\end{lemma}
\begin{proof}
If such an ordering exists, the $m$ moves end in $T$, so $\de\le|M|=\dpc\le\de$. Conversely, consider a sequence of $|M|$ moves from $S$ to $T$. Each move changes one parent (\autoref{lem:effect}) and every $v\in M$ must change its parent, so every $v\in M$ moves exactly once and no other mutation moves. The single move of $v$ must take it directly from $s(v)$ to $t(v)$.
\end{proof}

\subparagraph{The four cases, precisely.}
Fix a schedule as in \autoref{lem:schedule} and write $\pa(x)$ for the parent of $x$ in the current tree. Every mutation moves at most once. Hence $\pa(x)=s(x)$ at all times if $x\notin M$, and if $x\in M$, then $\pa(x)=s(x)$ before the jump of $x$ and $\pa(x)=t(x)$ after it. Before its own jump, $v$ is a child of $s(v)$, so its jump is an elementary move in exactly two situations:
\begin{itemize}
\item \emph{up-jump}: $s(v)\neq0$ and $\pa(s(v))=t(v)$, that is, $t(v)$ is the grandparent of $v$;
\item \emph{down-jump}: $t(v)\neq0$ and $\pa(t(v))=s(v)$, that is, $t(v)$ is a sibling of $v$.
\end{itemize}
Both conditions have the form $\pa(w)=r$ for one witness $w$ and one required parent $r$:
\[
\text{up-jump: } w=s(v),\ r=t(v);\qquad\qquad \text{down-jump: } w=t(v),\ r=s(v).
\]
The witness is the parent of $v$ in $S$ for an up-jump and in $T$ for a down-jump, and $r$ is the \emph{other} parent of $v$. So $r$ is determined by $v$ alone, not by $w$; the cases below compare it with the two parents $s(w)$ and $t(w)$ of the witness. By the rule above, at the moment when $v$ jumps the condition $\pa(w)=r$ is true
\begin{itemize}
\item always, if $w\notin M$ and $s(w)=r$ ($w$ is a child of $r$ and never moves);
\item if and only if $v$ jumps before $w$, when $w\in M$ and $s(w)=r$ ($w$ is a child of $r$ now, but leaves $r$ with its jump);
\item if and only if $v$ jumps after $w$, when $w\in M$ and $t(w)=r$ ($w$ is not a child of $r$ yet, and arrives there with its jump);
\item never, otherwise.
\end{itemize}
The second and third cases exclude each other, because $s(w)\neq t(w)$ for $w\in M$. \autoref{def:atoms} records these outcomes. A witness of the first kind makes $v$ free. A witness of the second kind is $\beta(v)$, and one of the third kind is $\alpha(v)$. In each pair of formulas below, the first one is the up-jump ($w=s(v)$) and the second one is the down-jump ($w=t(v)$).

\begin{definition}\label{def:atoms}
Let $v\in M$. We say that $v$ is \emph{free} if $s(v)\neq0$, $s(v)\notin M$ and $s(s(v))=t(v)$, or if $t(v)\neq0$, $t(v)\notin M$ and $s(t(v))=s(v)$. Otherwise we define
\begin{align*}
\beta(v)&=s(v) \text{ if } s(v)\in M,\ s(s(v))=t(v); & \beta(v)&=t(v) \text{ if } t(v)\in M,\ s(t(v))=s(v);\\
\alpha(v)&=s(v) \text{ if } s(v)\in M,\ t(s(v))=t(v); & \alpha(v)&=t(v) \text{ if } t(v)\in M,\ t(t(v))=s(v),
\end{align*}
and $\beta(v)$, $\alpha(v)$ are undefined otherwise. A non-free $v$ for which neither is defined is \emph{dead}.
\end{definition}

Thus $\beta(v)$ is the neighbour of $v$ such that $v$ must jump \emph{before} $\beta(v)$, and $\alpha(v)$ is the neighbour of $v$ such that $v$ must jump \emph{after} $\alpha(v)$. In both cases the neighbour is the witness of the jump of $v$, and the adjacency it provides exists on only one side of its own jump.

In \autoref{fig:jumps}(c), $S$ contains $g\to u\to v$ and $T$ contains $g\to v\to u$, so $s(v)=u$, $t(v)=g$, $s(u)=g$ and $t(u)=v$. For $v$, the up-jump has witness $w=s(v)=u$ and needs $\pa(u)=t(v)=g$. This holds before $u$ jumps ($s(u)=g$) but not after ($t(u)=v$), so $v$ must jump before $u$, that is, $\beta(v)=u$. For $u$, the down-jump has witness $w=t(u)=v$ and needs $\pa(v)=s(u)=g$. This fails before $v$ jumps ($s(v)=u$) and holds after ($t(v)=g$), so $u$ must jump after $v$, that is, $\alpha(u)=v$. (The other jump types give nothing here: an up-jump of $u$ would need $\pa(g)=v$, and a down-jump of $v$ would need $\pa(g)=u$.) So $\beta(v)=u$ and $\alpha(u)=v$ state the same requirement, ``$v$ before $u$'', from the point of view of each endpoint. In the order $v,u$, the mutation $v$ climbs to $g$ while $u$ is still a child of $g$, and then $u$ descends into its new sibling $v$.

\autoref{fig:beforeafter} separates the two kinds of constraint. In both pairs, $v$ makes a down-jump into $w$, so the witness is $w$ and the required parent is $r=s(v)=p$. In (i), $s(w)=p$: $w$ is a sibling of $v$ now and stops being one when it jumps, so $\beta(v)=w$. In (ii), $t(w)=p$: $w$ lies below $q$ and becomes a sibling of $v$ only by its jump, so $\alpha(v)=w$. In both cases the wrong order forces a detour of $v$, and the pair needs three moves instead of $|M|=2$.

\begin{figure}[t]
\centering
\begin{tikzpicture}[font=\scriptsize,
  m/.style={circle,draw,inner sep=0.5pt,minimum size=11pt},
  r/.style={m,fill=red!25},
  mv/.style={-{Stealth[length=4pt]},thick,gray},
  lbl/.style={font=\scriptsize\bfseries},
  ttl/.style={font=\scriptsize,anchor=south}]
\node[lbl,anchor=west] at (-1.0,0.75) {(i) $v$ before $w$: $\beta(v)=w$};
\begin{scope}
\node[ttl] at (0,0.2) {$S$};
\node[m] (o) at (0,0) {$0$}; \node[m] (p) at (0,-0.7) {$p$}; \node[r] (v) at (-0.35,-1.4) {$v$}; \node[r] (w) at (0.35,-1.4) {$w$};
\draw (o)--(p)--(v) (p)--(w);
\end{scope}
\draw[mv] (0.6,-0.7)--(1.4,-0.7) node[midway,above,align=center] {1: $v$ down\\into $w$};
\begin{scope}[xshift=2.0cm]
\node[m] (o) at (0,0) {$0$}; \node[m] (p) at (0,-0.7) {$p$}; \node[r] (w) at (0,-1.4) {$w$}; \node[r] (v) at (0,-2.1) {$v$};
\draw (o)--(p)--(w)--(v);
\end{scope}
\draw[mv] (2.6,-0.7)--(3.4,-0.7) node[midway,above,align=center] {2: $w$ up\\to $0$};
\begin{scope}[xshift=4.0cm]
\node[ttl] at (0,0.2) {$T$};
\node[m] (o) at (0,0) {$0$}; \node[m] (p) at (-0.35,-0.7) {$p$}; \node[r] (w) at (0.35,-0.7) {$w$}; \node[r] (v) at (0.35,-1.4) {$v$};
\draw (o)--(p) (o)--(w)--(v);
\end{scope}
\node[align=left,anchor=west] at (5.0,-0.9) {$r=s(v)=p=s(w)$: $w$ is a sibling of $v$ now\\ and stops being one when it jumps.\\ Wrong order: after $w$ moves up, $v$ needs\\ a detour through $0$ (3 moves).};
\begin{scope}[yshift=-3.6cm]
\node[lbl,anchor=west] at (-1.0,0.75) {(ii) $v$ after $w$: $\alpha(v)=w$};
\begin{scope}
\node[ttl] at (0,0.2) {$S$};
\node[m] (o) at (0,0) {$0$}; \node[m] (p) at (0,-0.7) {$p$}; \node[r] (v) at (-0.35,-1.4) {$v$}; \node[m] (q) at (0.35,-1.4) {$q$}; \node[r] (w) at (0.35,-2.1) {$w$};
\draw (o)--(p)--(v) (p)--(q)--(w);
\end{scope}
\draw[mv] (0.6,-0.7)--(1.4,-0.7) node[midway,above,align=center] {1: $w$ up\\to $p$};
\begin{scope}[xshift=2.0cm]
\node[m] (o) at (0,0) {$0$}; \node[m] (p) at (0,-0.7) {$p$}; \node[r] (v) at (-0.5,-1.4) {$v$}; \node[m] (q) at (0,-1.4) {$q$}; \node[r] (w) at (0.5,-1.4) {$w$};
\draw (o)--(p)--(v) (p)--(q) (p)--(w);
\end{scope}
\draw[mv] (2.6,-0.7)--(3.4,-0.7) node[midway,above,align=center] {2: $v$ down\\into $w$};
\begin{scope}[xshift=4.0cm]
\node[ttl] at (0,0.2) {$T$};
\node[m] (o) at (0,0) {$0$}; \node[m] (p) at (0,-0.7) {$p$}; \node[m] (q) at (-0.35,-1.4) {$q$}; \node[r] (w) at (0.35,-1.4) {$w$}; \node[r] (v) at (0.35,-2.1) {$v$};
\draw (o)--(p)--(q) (p)--(w)--(v);
\end{scope}
\node[align=left,anchor=west] at (5.0,-0.9) {$r=s(v)=p=t(w)$: $w$ becomes a sibling of $v$\\ only by its jump.\\ Wrong order: $v$ cannot enter $w$ below $q$\\ directly and needs a detour through $q$ (3 moves).};
\end{scope}
\end{tikzpicture}
\caption{The two kinds of order constraint, isolated, each with its optimal schedule of $|M|=2$ moves (misplaced mutations shaded). In both pairs, $v$ makes a down-jump into $w$, whose required parent is $r=s(v)=p$. Whether $v$ must go before or after $w$ depends only on whether $p$ is the parent of $w$ in $S$ (i) or in $T$ (ii).}\label{fig:beforeafter}
\end{figure}

\begin{lemma}\label{lem:atoms}
$\beta(v)$ and $\alpha(v)$ are well defined. In a schedule as in \autoref{lem:schedule}, the move of a non-free $v\in M$ is valid if and only if $\beta(v)$ is defined and moves after $v$, or $\alpha(v)$ is defined and moves before $v$. A free $v$ can move at any time.
\end{lemma}
\begin{proof}
If both clauses defining $\beta(v)$ held, $S$ would contain the edges $t(v)\to s(v)$ and $s(v)\to t(v)$, a cycle. If both clauses for $\alpha(v)$ held, $T$ would contain $t(v)\to s(v)$ and $s(v)\to t(v)$. The condition $\pa(s(v))=t(v)$ holds at time $i$ exactly when $s(v)\ne0$ and either $s(v)\notin M$ with $s(s(v))=t(v)$ (always, so $v$ is free), or $s(v)\in M$ has not moved yet and $s(s(v))=t(v)$ (this is $\beta(v)=s(v)$, moving after $v$), or $s(v)\in M$ has already moved and $t(s(v))=t(v)$ (this is $\alpha(v)=s(v)$, moving before $v$). The condition $\pa(t(v))=s(v)$ is analysed in the same way with the roles of $s(v)$ and $t(v)$ exchanged.
\end{proof}

Let $N$ be the set of non-free mutations of $M$. A \emph{selection} chooses for every $v\in N$ one defined atom. It induces a digraph $D$ on $M$ with an arc $v\to\beta(v)$ if $v$ chose $\beta$ (``$v$ before $\beta(v)$'') and an arc $\alpha(v)\to v$ if $v$ chose $\alpha$ (``$\alpha(v)$ before $v$''). By \autoref{lem:schedule} and \autoref{lem:atoms}, $\de(S,T)=|M|$ if and only if there is no dead mutation and some selection gives an acyclic $D$, since any topological order of $D$ is then a valid schedule. Let $B$ and $A$ be the functional digraphs $v\mapsto\beta(v)$ and $v\mapsto\alpha(v)$ on $N$.

\begin{lemma}\label{lem:uniform}
Every directed cycle of $D$ is a cycle of $B$ all of whose vertices chose $\beta$, or a cycle of $A$, traversed backwards, all of whose vertices chose $\alpha$.
\end{lemma}
\begin{proof}
Call the vertex that chose an arc its \emph{owner}. Every arc has exactly one owner, which is one of its endpoints, and every vertex owns at most one arc. A directed cycle $Z$ with $k\ge2$ vertices (there are no loops, as $\beta(v),\alpha(v)\neq v$) has $k$ arcs with $k$ distinct owners on $Z$, so every vertex of $Z$ owns exactly one arc of $Z$. If some $z\in Z$ owns its outgoing arc $z\to z'$, then $z'$ does not own that arc and must own its own outgoing arc. Inductively every vertex owns its outgoing arc, so all arcs are $v\to\beta(v)$ and $Z$ is a cycle of $B$. Otherwise every vertex owns its incoming arc, so all arcs are $\alpha(v)\to v$ and $Z$ is a cycle of $A$.
\end{proof}

\begin{lemma}\label{lem:acyclic}
The digraphs $B$ and $A$ are acyclic.
\end{lemma}
\begin{proof}
Consider $B$. If $\beta(v)=s(v)$, then $\beta(v)$ is the parent of $v$ in $S$, so $\dep_S(\beta(v))=\dep_S(v)-1$. If $\beta(v)=t(v)$, then $s(\beta(v))=s(v)$, so $\beta(v)$ is a sibling of $v$ in $S$ and $\dep_S(\beta(v))=\dep_S(v)$. Along a directed cycle the depth in $S$ returns to its initial value and never increases, so every arc of the cycle is of the second kind. But then $\beta(v)=t(v)$ for every $v$ on the cycle, which is a cycle of the parent function of $T$, a contradiction. The argument for $A$ is the mirror image: if $\alpha(v)=t(v)$ then $\alpha(v)$ is the parent of $v$ in $T$, if $\alpha(v)=s(v)$ then $\alpha(v)$ is a sibling of $v$ in $T$, and a cycle would consist of arcs $\alpha(v)=s(v)$ only, a cycle of the parent function of $S$.
\end{proof}

\begin{theorem}\label{thm:tight}
$\de(S,T)=\dpc(S,T)$ if and only if no mutation is dead. This can be decided in $O(n)$ time, and a schedule of $\dpc(S,T)$ moves can be computed within the same time bound when it exists.
\end{theorem}
\begin{proof}
If $v$ is dead, then by \autoref{lem:atoms} the jump of $v$ is never valid, so no schedule as in \autoref{lem:schedule} exists and $\de>\dpc$. If no mutation is dead, choose any defined atom for every $v\in N$. By \autoref{lem:uniform} and \autoref{lem:acyclic}, the digraph $D$ is acyclic, and a topological order of $D$ is a valid schedule by \autoref{lem:atoms}. The atoms are computed in $O(1)$ time per mutation from the parent functions, and $D$ has at most $|M|$ arcs, so the test and the topological sort take $O(n)$ time.
\end{proof}

\begin{algorithm}[t]
\caption{Linear-time test for $\de(S,T)=\dpc(S,T)$, with a schedule (\autoref{thm:tight}).}\label{alg:tight}
\KwIn{parent functions $s=\pa_S$, $t=\pa_T$ on $L$}
$M\gets\{v: s(v)\ne t(v)\}$; $D\gets$ empty digraph on $M$\;
\ForEach{$v\in M$}{
  \lIf{$v$ is free}{\textbf{continue}}
  \lIf{$\beta(v)$ is defined}{add the arc $v\to\beta(v)$ to $D$}
  \lElseIf{$\alpha(v)$ is defined}{add the arc $\alpha(v)\to v$ to $D$}
  \lElse{\Return ``$\de>\dpc$'' ($v$ is dead)}
}
\Return a topological order of $D$, i.e.\ a schedule of $\dpc$ moves\;
\end{algorithm}

\autoref{alg:tight} summarizes the test. When the test succeeds, the optimal transformation is \emph{exactly} a link-and-cut sequence in which every reattachment happens to be local. When it fails, $\de>\dpc$.

\subparagraph{A lower bound on the excess.}
Dead mutations give a quantitative bound that holds for every pair of trees. For a dead $v$ let $W(v)=\{v,s(v),t(v)\}\setminus\{0\}$ be its \emph{detour set}. By \autoref{lem:dead} below, some mutation of $W(v)$ makes more moves than in a schedule of $\dpc$ moves. A misplaced mutation that does so makes at least one extra move, and a mutation that is not misplaced makes at least two (it leaves its parent and comes back). Accordingly, every mutation $x$ gets the weight $\omega(x)=1$ if $x\in M$ and $\omega(x)=2$ otherwise. A set $H$ of mutations \emph{hits} $W(v)$ if $H\cap W(v)\neq\emptyset$, and its weight is $\omega(H)=\sum_{x\in H}\omega(x)$. Let $\tau_\omega(S,T)$ be the minimum weight of a set $H$ that hits the detour set of every dead mutation ($\tau_\omega=0$ if there is no dead mutation). In \autoref{fig:jumps}(d), the only dead mutation is $x$, $W(x)=\{x,a,b\}$ has weights $1,2,2$, and $\tau_\omega=1$.

\begin{lemma}\label{lem:dead}
Let $v$ be dead and consider any sequence of moves from $S$ to $T$. Then some $x\in W(v)$ moves at least twice. Moreover, if $x\neq v$, then at some moment the parent of $x$ is neither its parent in $S$ nor its parent in $T$.
\end{lemma}
\begin{proof}
Suppose $v$ moves only once. Since $v\in M$, this move takes $v$ from $s(v)$ directly to $t(v)$, so at that moment $\pa(s(v))=t(v)$ with $s(v)\neq0$, or $\pa(t(v))=s(v)$ with $t(v)\neq0$. In the first case, if the current parent of $s(v)$ were $s(s(v))$ or $t(s(v))$, then $v$ would be free (if $s(v)\notin M$) or $\beta(v)$ or $\alpha(v)$ would be defined (if $s(v)\in M$), contradicting that $v$ is dead. Hence $s(v)$ sits at a parent different from both its initial and its final parent. It has therefore moved at least once and must move at least once more. The second case is symmetric with $t(v)$.
\end{proof}

\begin{theorem}\label{thm:hitting}
For all $S,T\in\T_n$, $\de(S,T)\ge\dpc(S,T)+\tau_\omega(S,T)$.
\end{theorem}
\begin{proof}
Fix an optimal sequence and let $m(x)$ be the number of moves of $x$. Then $\de-\dpc=\sum_{x\in M}(m(x)-1)+\sum_{x\notin M}m(x)$, and all terms are nonnegative. Let $X=\{x:m(x)\ge2\}$. A mutation $x\in X\cap M$ contributes at least $1=\omega(x)$, and $x\in X\setminus M$ contributes at least $2=\omega(x)$. By \autoref{lem:dead}, $X$ hits every detour set, so $\de-\dpc\ge\omega(X)\ge\tau_\omega$.
\end{proof}

In \autoref{fig:jumps}(d) the bound gives $\de\ge1+1=2$, which is exact. Computing $\tau_\omega$ is a weighted hitting-set problem with sets of size at most three, which is \NP-hard in general. When the detour sets are few, $\tau_\omega$ can be computed exactly by enumeration. Since \autoref{thm:hitting} is a lower bound, any lower bound on $\tau_\omega$ may be used instead, for instance a packing of pairwise disjoint detour sets, each contributing the smallest weight of its elements.

\subparagraph{Column accounting and packaging.}
Since $\de\ge\dn$, \autoref{thm:lb} applies to $\de$; in particular $\de\ge\mathrm N$. Writing $\mathrm{ops}(z)$ for the number of operations on column $z$ in a sequence, every elementary move is exactly one operation on one column, so
\begin{equation}\label{eq:columns}
\de(S,T)=\min\ \sum_{z\in L}\mathrm{ops}(z),
\end{equation}
where the minimum is over all sequences of elementary moves from $S$ to $T$. For elementary moves the column view is sharper than for \nnip, because a column that only gains (or only loses) can be served by a single operation only if the whole set arrives (or leaves) as \emph{one} subtree.

\begin{lemma}[packaging]\label{lem:packaging}
Let $z$ satisfy $L_z=\emptyset\neq G_z$ and put $X=G_z$ (or $G_z=\emptyset\neq L_z$ and $X=L_z$). In every sequence from $S$ to $T$, either $\mathrm{ops}(z)\ge2$, or there are a moment and a $w\in X$ with $\cl(w)=X$ at that moment.
\end{lemma}
\begin{proof}
Suppose $\mathrm{ops}(z)=1$ and $X=G_z$. Before the single operation $\cl(z)=\cl_S(z)$; the operation adds the cluster $\cl(w)$ of a sibling $w$, which is disjoint from $\cl(z)$, and afterwards $\cl(z)=\cl_T(z)$. Hence $\cl(w)=G_z$ and $w\in G_z$. The case $X=L_z$ is symmetric: the single operation removes the cluster of a child $c$, and $\cl(c)=L_z$.
\end{proof}

If $X\notin\{\cl_S(w),\cl_T(w)\}$, a column $w$ that passes through the cluster $X$ needs at least two operations (one to reach $X$, one to leave it). This costs at least one operation beyond $\mathrm{need}(w)$ whenever $\mathrm{need}(w)\le1$. For every column $z$ as in \autoref{lem:packaging} such that no $w\in X$ has $\mathrm{need}(w)=2$ or $X\in\{\cl_S(w),\cl_T(w)\}$, define the \emph{packaging witness set} $P(z)=\{z\}\cup X$, and let $\tau_c$ be the minimum size of a set of mutations intersecting all packaging witness sets.

\begin{theorem}[column bound]\label{thm:column}
$\de(S,T)\ge\sum_{z\in L}\mathrm{need}(z)+\tau_c(S,T)$.
\end{theorem}
\begin{proof}
Fix an optimal sequence. By~\eqref{eq:columns}, $\de-\sum_z\mathrm{need}(z)=\sum_z(\mathrm{ops}(z)-\mathrm{need}(z))$, and each term is nonnegative. Let $Y$ be the set of columns with $\mathrm{ops}(z)>\mathrm{need}(z)$. By \autoref{lem:packaging} and the remark following it, $Y$ intersects every packaging witness set: either $z\in Y$, or some $w\in X$ with $\mathrm{need}(w)\le1$ passes through the intermediate cluster $X$ and hence $w\in Y$. Thus $\de-\sum_z\mathrm{need}(z)\ge|Y|\ge\tau_c$.
\end{proof}

We write $\LB^+_1=\max\{\LB,\ \dpc+\tau_\omega,\ \mathrm N+\tau_c\}$. The three parts count the same moves by moved mutation (rows), by changed cluster (columns) and by depth, which is why they complement each other. The bound is not always tight: one extra move of a mutation may serve several dead mutations, and a column may be crossed in both directions, and neither effect is captured by a hitting-set argument.

\section{The elementary distance is NP-complete}\label{sec:complexity}

The decision problem \textsc{nni+$_1$ Distance} asks, given $S,T\in\T_n$ and an integer $k$, whether $\de(S,T)\le k$. It is in \NP: a yes-instance has a certificate of at most $2n-2$ moves (\autoref{sec:tight}), each checkable in $O(n)$ time.

\autoref{thm:tight} shows that deciding $\de=\dpc$ is easy, but the excess $\de-\dpc$ is not bounded by a constant; the reduction below produces instances with excess $\tau(G)$. The source of difficulty is that the misplaced mutations impose an order on each other (a mutation can enter its new parent only while that parent is at the right place), these orders can form cycles, and a cycle costs one extra move to break. One extra move can break many cycles at once. This is exactly the structure of \textsc{Vertex Cover}, and we now turn it into a reduction.

\begin{theorem}\label{thm:np}
\textsc{nni+$_1$ Distance} is \NP-complete.
\end{theorem}

The rest of this section proves \autoref{thm:np} by a reduction from \textsc{Vertex Cover}~\cite{karp1972reducibility}. Routine case checks are deferred to \autoref{app:np}.

\subparagraph{Idea of the reduction.}
The construction has three ingredients.
\begin{itemize}
\item A \emph{mediator} turns an order requirement ``$p$ moves before $q$'' into a cost. It is a leaf that has to travel from $p$ to $q$. It needs two moves if $p$ moves first and three otherwise (\autoref{fig:mediator}).
\item A \emph{vertex gadget} is a chain of such requirements that leads from an entry node $c^{\mathrm{in}}_u$ to an exit node $c^{\mathrm{out}}_u$. Paying one extra move for the mediator $y_u$ cuts the chain.
\item An \emph{edge gadget} links the exit of each endpoint to the entry of the other. Two adjacent vertices thus close a cycle of requirements (\autoref{fig:gadget}), which has to be cut somewhere.
\end{itemize}
Cutting a vertex chain cuts every cycle through that vertex, so the cheapest way to cut all cycles costs $\tau(G)$ extra moves. The upper bound is therefore easy. The work lies in the lower bound, which must exclude every cheaper trick, such as detours of other mutations, mediators travelling through other gadgets, or one extra move serving several cycles.

\subsection{The construction}\label{sec:construction}

We use the notation of \autoref{sec:tight}: $s(v)=\pa_S(v)$, $t(v)=\pa_T(v)$, $M$ is the set of misplaced mutations, and we write $C(v)=\{s(v),t(v)\}$ for the two \emph{canonical parents} of $v$.

\subparagraph{Precedence gadget.}
Let $p,q$ be misplaced mutations with $t(p)=s(q)$. A \emph{mediator} $\langle p,q\rangle$ is a new mutation $y$ that is a leaf in both trees, with $s(y)=p$ and $t(y)=q$. We call $p$ and $q$ the \emph{source} and \emph{target} of $y$, \emph{partners} of each other, and $t(p)=s(q)$ the \emph{port} of $y$. If $p$ and $q$ each move once, then $y$ can be served with two moves exactly when $p$ moves before $q$: after $p$ has arrived at the port, $y$ climbs from $p$ to the port, and it then descends into $q$, which is still a child of the port. If $q$ moves first, $y$ needs a third move. The mediator thus turns ``$p$ before $q$'' into a constraint that costs one extra move to violate.

\begin{figure}[t]
\centering
\begin{tikzpicture}[font=\scriptsize,
  m/.style={circle,draw,inner sep=0.5pt,minimum size=12pt},
  md/.style={rectangle,draw,fill=blue!15,inner sep=1.5pt,minimum size=9pt},
  j/.style={->,thick,red!70!black,densely dashed},
  lbl/.style={font=\scriptsize\bfseries}]
\begin{scope}
\node[lbl] at (0.9,0.55) {(i) start};
\node[m] (a) at (0,0) {$s(p)$}; \node[m] (P) at (1.8,0) {$P$};
\node[m,fill=red!20] (p) at (0,-0.9) {$p$}; \node[m,fill=red!20] (q) at (1.8,-0.9) {$q$};
\node[md] (y) at (0,-1.7) {$y$};
\draw (a)--(p)--(y) (P)--(q);
\draw[j] (p) to[bend left=20] node[above,font=\tiny]{$p$ jumps} (P);
\end{scope}
\begin{scope}[xshift=3.6cm]
\node[lbl] at (0.9,0.55) {(ii) $y$ climbs};
\node[m] (P) at (0.9,0) {$P$};
\node[m,fill=red!20] (p) at (0.3,-0.9) {$p$}; \node[m,fill=red!20] (q) at (1.5,-0.9) {$q$};
\node[md] (y) at (0.3,-1.7) {$y$};
\draw (P)--(p)--(y) (P)--(q);
\draw[j] (y) to[bend left=50] (P);
\end{scope}
\begin{scope}[xshift=6.6cm]
\node[lbl] at (0.9,0.55) {(iii) $y$ descends};
\node[m] (P) at (0.9,0) {$P$};
\node[m,fill=red!20] (p) at (0.2,-0.9) {$p$}; \node[md] (y) at (0.9,-0.9) {$y$}; \node[m,fill=red!20] (q) at (1.6,-0.9) {$q$};
\draw (P)--(p) (P)--(y) (P)--(q);
\draw[j] (y) to[bend right=40] (q);
\end{scope}
\begin{scope}[xshift=9.6cm]
\node[lbl] at (1.0,0.55) {(iv) wrong order};
\node[m] (P) at (0.4,0) {$P$}; \node[m] (tq) at (1.8,0) {$t(q)$};
\node[m,fill=red!20] (p) at (0.4,-0.9) {$p$}; \node[m,fill=red!20] (q) at (1.8,-0.9) {$q$};
\node[md] (y) at (0.4,-1.7) {$y$};
\draw (P)--(p)--(y) (tq)--(q);
\node[align=left,anchor=north west] at (-0.3,-2.0) {$q$ left $P$ before $p$ arrived:\\ $y$ needs a third move};
\end{scope}
\end{tikzpicture}
\caption{The mediator $y=\langle p,q\rangle$ with port $P=t(p)=s(q)$. If $p$ jumps to $P$ before $q$ leaves it, $y$ reaches $q$ in two moves (ii, iii) and then travels with $q$. If $q$ leaves first (iv), $y$ and $q$ are never siblings at the right moment, and $y$ needs an extra move. Mediators are drawn as squares.}\label{fig:mediator}
\end{figure}

\subparagraph{Vertex gadget.}
Let $G=(V,E)$ be a graph. For every $u\in V$ we add the mutations in \autoref{tab:gadget}. $S_u$ and $T_u$ are \emph{fixed} ($s=t$). The eight \emph{connectors} $U_u,p_u,q_u,c^1_u,c^{\mathrm{in}}_u,w_u,\bar w_u,c^{\mathrm{out}}_u$ are misplaced and are not mediators. Finally, there are four \emph{internal mediators}
\[
\langle w_u,p_u\rangle,\quad y_u=\langle p_u,q_u\rangle,\quad \langle q_u,c^1_u\rangle,\quad \langle c^1_u,\bar w_u\rangle .
\]
\subparagraph{Edge gadget.}
For every edge $\{a,c\}\in E$ we add the two \emph{arc mediators} $x_{ac}=\langle c^{\mathrm{out}}_a,c^{\mathrm{in}}_c\rangle$ and $x_{ca}=\langle c^{\mathrm{out}}_c,c^{\mathrm{in}}_a\rangle$, whose port is the root $0$.

\begin{table}[t]
\centering\small
\caption{The vertex gadget of $u$ (mediators omitted). Each connector moves once in a direct solution, and the comments give the order this imposes.}\label{tab:gadget}
\begin{tabular}{l l l l}
\toprule
mutation & $s(\cdot)$ & $t(\cdot)$ & a single direct move requires\\
\midrule
$S_u$ & $0$ & $0$ & (fixed)\\
$T_u$ & $S_u$ & $S_u$ & (fixed)\\
$U_u$ & $T_u$ & $S_u$ & --- \\
$p_u$ & $S_u$ & $T_u$ & --- \\
$q_u$ & $T_u$ & $U_u$ & $U_u$ has not moved: $q_u$ before $U_u$\\
$c^1_u$ & $U_u$ & $S_u$ & $U_u$ has moved: $U_u$ before $c^1_u$\\
$c^{\mathrm{in}}_u$ & $0$ & $S_u$ & --- \\
$w_u$ & $c^{\mathrm{in}}_u$ & $S_u$ & $c^{\mathrm{in}}_u$ has moved: $c^{\mathrm{in}}_u$ before $w_u$\\
$\bar w_u$ & $S_u$ & $c^{\mathrm{out}}_u$ & $c^{\mathrm{out}}_u$ has not moved: $\bar w_u$ before $c^{\mathrm{out}}_u$\\
$c^{\mathrm{out}}_u$ & $S_u$ & $0$ & --- \\
\bottomrule
\end{tabular}
\end{table}

\begin{figure}[t]
\centering
\begin{tikzpicture}[font=\scriptsize,>=Stealth,
  n/.style={draw,rounded corners=2pt,inner sep=2pt,minimum height=11pt},
  med/.style={->,thick,blue!70!black},
  nat/.style={->,thick,densely dashed,gray!70!black},
  arc/.style={->,thick,red!70!black}]
\begin{scope}
\node[n] (ia) at (0,0) {$c^{\mathrm{in}}_a$};
\node[n] (wa) at (1.3,0) {$w_a$};
\node[n] (pa) at (2.5,0) {$p_a$};
\node[n] (qa) at (3.7,0) {$q_a$};
\node[n] (Ua) at (4.9,0) {$U_a$};
\node[n] (1a) at (6.1,0) {$c^1_a$};
\node[n] (ba) at (7.3,0) {$\bar w_a$};
\node[n] (oa) at (8.6,0) {$c^{\mathrm{out}}_a$};
\draw[nat] (ia)--(wa); \draw[med] (wa)--(pa); \draw[med] (pa)--node[above]{$y_a$}(qa);
\draw[nat] (qa)--(Ua); \draw[nat] (Ua)--(1a); \draw[med] (qa) to[bend left=25] (1a);
\draw[med] (1a)--(ba); \draw[nat] (ba)--(oa);
\end{scope}
\begin{scope}[yshift=-1.4cm]
\node[n] (ic) at (8.6,0) {$c^{\mathrm{in}}_c$};
\node[n] (wc) at (7.3,0) {$w_c$};
\node[n] (pc) at (6.1,0) {$p_c$};
\node[n] (qc) at (4.9,0) {$q_c$};
\node[n] (Uc) at (3.7,0) {$U_c$};
\node[n] (1c) at (2.5,0) {$c^1_c$};
\node[n] (bc) at (1.3,0) {$\bar w_c$};
\node[n] (oc) at (0,0) {$c^{\mathrm{out}}_c$};
\draw[nat] (ic)--(wc); \draw[med] (wc)--(pc); \draw[med] (pc)--node[below]{$y_c$}(qc);
\draw[nat] (qc)--(Uc); \draw[nat] (Uc)--(1c); \draw[med] (qc) to[bend right=25] (1c);
\draw[med] (1c)--(bc); \draw[nat] (bc)--(oc);
\end{scope}
\draw[arc] (oa) -- node[right]{$x_{ac}$} (ic);
\draw[arc] (oc) -- node[left]{$x_{ca}$} (ia);
\end{tikzpicture}
\caption{The order constraints of the gadgets of two adjacent vertices $a,c$, if every connector moves once. Blue: internal mediators; dashed: orders forced by the gadget itself (\autoref{tab:gadget}); red: arc mediators. The two gadgets and the two arcs form a cycle, which must be broken by an extra move. Breaking $y_u$ costs one move and breaks every cycle through $u$.}\label{fig:gadget}
\end{figure}

\autoref{fig:gadgettrees} shows the gadget of $u$ in both trees; in particular $S$ and $T$ are trees. The instance has $n=14|V|+2|E|$ mutations, $\dpc=|M|=12|V|+2|E|$, and $4|V|+2|E|$ mediators. Set
\[
K=\dpc+\#\text{mediators}=16|V|+4|E|.
\]

\begin{figure}[t]
\centering
\begin{tikzpicture}[font=\scriptsize,
  f/.style={circle,draw,inner sep=0.3pt,minimum size=13pt},
  c/.style={circle,draw,fill=red!20,inner sep=0.3pt,minimum size=13pt},
  lbl/.style={font=\small\bfseries}]
\begin{scope}
\node[lbl] at (1.6,0.6) {$S$};
\node[f] (r) at (1.9,0) {$0$};
\node[f] (Su) at (1.0,-0.85) {$S_u$}; \node[c] (ci) at (3.2,-0.85) {$c^{\mathrm{in}}_u$}; \node[c] (w) at (3.2,-1.75) {$w_u$};
\node[f] (Tu) at (-0.3,-1.75) {$T_u$}; \node[c] (p) at (0.55,-1.75) {$p_u$}; \node[c] (wb) at (1.4,-1.75) {$\bar w_u$}; \node[c] (co) at (2.25,-1.75) {$c^{\mathrm{out}}_u$};
\node[c] (U) at (-0.75,-2.65) {$U_u$}; \node[c] (q) at (0.15,-2.65) {$q_u$};
\node[c] (c1) at (-0.75,-3.55) {$c^1_u$};
\draw (r)--(Su) (r)--(ci)--(w) (Su)--(Tu) (Su)--(p) (Su)--(wb) (Su)--(co) (Tu)--(U)--(c1) (Tu)--(q);
\end{scope}
\begin{scope}[xshift=6.3cm]
\node[lbl] at (1.6,0.6) {$T$};
\node[f] (r) at (1.9,0) {$0$};
\node[f] (Su) at (1.0,-0.85) {$S_u$}; \node[c] (co) at (3.4,-0.85) {$c^{\mathrm{out}}_u$}; \node[c] (wb) at (3.4,-1.75) {$\bar w_u$};
\node[f] (Tu) at (-0.6,-1.75) {$T_u$}; \node[c] (U) at (0.2,-1.75) {$U_u$}; \node[c] (c1) at (1.0,-1.75) {$c^1_u$}; \node[c] (ci) at (1.8,-1.75) {$c^{\mathrm{in}}_u$}; \node[c] (w) at (2.6,-1.75) {$w_u$};
\node[c] (p) at (-0.6,-2.65) {$p_u$}; \node[c] (q) at (0.2,-2.65) {$q_u$};
\draw (r)--(Su) (r)--(co)--(wb) (Su)--(Tu) (Su)--(U) (Su)--(c1) (Su)--(ci) (Su)--(w) (Tu)--(p) (U)--(q);
\end{scope}
\end{tikzpicture}
\caption{The vertex gadget of $u$ in $S$ and in $T$ (white: fixed; red: connectors; mediators omitted). In $S$, the four internal mediators hang below their sources $w_u,p_u,q_u,c^1_u$; in $T$, below their targets $p_u,q_u,c^1_u,\bar w_u$. Arc mediators hang below $c^{\mathrm{out}}_u$ in $S$ and below $c^{\mathrm{in}}_u$ in $T$.}\label{fig:gadgettrees}
\end{figure}

\autoref{fig:gadget} shows the intended order. Following the chain of a gadget and the arcs of the edges, two adjacent vertices produce a cycle of order constraints. We will show that $\de(S,T)=K+\tau(G)$, where $\tau(G)$ is the vertex cover number. Together with membership in \NP this proves \autoref{thm:np}.

\begin{lemma}[upper bound]\label{lem:np-ub}
For every vertex cover $C$ of $G$, $\de(S,T)\le K+|C|$.
\end{lemma}
\begin{proof}
For $u\in C$, move $y_u$ three times before anything else: up to $S_u$, down into $T_u$, and down into $q_u$, which is still a child of $T_u$. Now consider the digraph on the connectors with arcs $c^{\mathrm{in}}_u\to w_u\to p_u\to q_u\to U_u\to c^1_u\to\bar w_u\to c^{\mathrm{out}}_u$ for every $u$, except the arc $p_u\to q_u$ when $u\in C$, and $c^{\mathrm{out}}_a\to c^{\mathrm{in}}_c$ for every ordered pair with $\{a,c\}\in E$. Every cycle of this digraph passes through the gadgets of two adjacent vertices, hence through some $u\in C$, whose chain is cut. So the digraph is acyclic. Move the connectors once each, directly to their targets, in a topological order. The direct moves are legal by \autoref{tab:gadget}. Each remaining mediator $\langle p,q\rangle$ moves twice, right after its source $p$ has moved: up to the port, then down into $q$, which has not moved yet. The total is $\dpc+\#\text{mediators}+|C|$.
\end{proof}

\subparagraph{An example.}
\autoref{fig:instance} shows the complete instance for the graph $G$ with $V=\{a,b,c,d\}$ and $E=\{ab,bc,cd,bd\}$. The pairs $\{a,c\}$ and $\{a,d\}$ are not adjacent, and the minimum vertex covers are $\{b,c\}$ and $\{b,d\}$, so $\tau(G)=2$. The instance has $n=64$ mutations, $|M|=56$, $24$ mediators and $K=80$. \autoref{tab:solution} gives an optimal sequence of $K+2=82$ moves, built as in \autoref{lem:np-ub} from the cover $C=\{b,c\}$.
\begin{itemize}
\item Moves 1--6 cut the chains of $b$ and $c$: $y_b$ and $y_c$ make three moves each before $p_b,q_b,p_c,q_c$ move. These are the only extra moves.
\item The second halves of these chains ($q_u$ to $\bar w_u$, moves 7--22) and their exits $c^{\mathrm{out}}_b,c^{\mathrm{out}}_c$ (moves 23--34) go first. The arc mediators leaving $b$ and $c$ then reach the entries of their neighbours while those are still children of $0$.
\item The chains of $a$ and $d$ stay intact and are processed completely (moves 35--72). This is possible only because $a$ and $d$ are not adjacent, so no arc joins their chains.
\item The first halves of the chains of $b$ and $c$ ($c^{\mathrm{in}}_u$, $w_u$, $p_u$) come last (moves 73--82).
\end{itemize}
Any extra move breaks some cycle of \autoref{fig:gadget}, but a cut of $y_u$ is the most economical one, because it breaks every cycle through $u$ at once. For comparison, paying the extra moves on arc mediators instead breaks only one cycle per arc and costs $|E|=4$ extra moves ($84$ in total). A set of vertices that is not a cover, such as $\{a,d\}$, leaves the cycle through the chains of $b$ and $c$ uncut.

\begin{figure}[t]
\centering
\begin{tikzpicture}[font=\small,vx/.style={circle,draw,minimum size=15pt},cv/.style={vx,fill=orange!40}]
\node[vx] (a) at (0,0) {$a$}; \node[cv] (b) at (1.4,0) {$b$}; \node[cv] (c) at (2.7,0.6) {$c$}; \node[vx] (d) at (2.7,-0.6) {$d$};
\draw[thick] (a)--(b)--(c)--(d)--(b);
\node[anchor=west] at (3.4,0) {$G$; minimum vertex cover $C=\{b,c\}$ (orange)};
\end{tikzpicture}

\medskip
\resizebox{\linewidth}{!}{\begin{tikzpicture}[font=\scriptsize,fx/.style={circle,draw,fill=white,inner sep=0.2pt,minimum size=13pt},cn/.style={circle,draw,fill=red!20,inner sep=0.2pt,minimum size=13pt},md/.style={rectangle,draw,fill=blue!15,inner sep=1pt,minimum size=10pt}]
\node[font=\small\bfseries] at (8.650,0.6) {S};
\draw (8.650,0.000)--(1.000,-0.620);
\draw (1.000,-0.620)--(0.750,-1.240);
\draw (0.750,-1.240)--(0.500,-1.860);
\draw (1.000,-0.620)--(1.500,-1.240);
\draw (0.750,-1.240)--(1.000,-1.860);
\draw (0.500,-1.860)--(0.500,-2.480);
\draw (8.650,0.000)--(2.900,-0.620);
\draw (1.000,-0.620)--(0.000,-1.240);
\draw (2.900,-0.620)--(2.900,-1.240);
\draw (1.000,-0.620)--(2.000,-1.240);
\draw (1.500,-1.240)--(1.500,-1.860);
\draw (1.000,-1.860)--(1.000,-2.480);
\draw (0.500,-2.480)--(0.500,-3.100);
\draw (2.900,-1.240)--(2.900,-1.860);
\draw (8.650,0.000)--(5.550,-0.620);
\draw (5.550,-0.620)--(5.550,-1.240);
\draw (5.550,-1.240)--(5.300,-1.860);
\draw (5.550,-0.620)--(6.300,-1.240);
\draw (5.550,-1.240)--(5.800,-1.860);
\draw (5.300,-1.860)--(5.300,-2.480);
\draw (8.650,0.000)--(7.700,-0.620);
\draw (5.550,-0.620)--(4.300,-1.240);
\draw (7.700,-0.620)--(7.700,-1.240);
\draw (5.550,-0.620)--(6.800,-1.240);
\draw (6.300,-1.240)--(6.300,-1.860);
\draw (5.800,-1.860)--(5.800,-2.480);
\draw (5.300,-2.480)--(5.300,-3.100);
\draw (7.700,-1.240)--(7.700,-1.860);
\draw (8.650,0.000)--(9.975,-0.620);
\draw (9.975,-0.620)--(9.850,-1.240);
\draw (9.850,-1.240)--(9.600,-1.860);
\draw (9.975,-0.620)--(10.600,-1.240);
\draw (9.850,-1.240)--(10.100,-1.860);
\draw (9.600,-1.860)--(9.600,-2.480);
\draw (8.650,0.000)--(12.000,-0.620);
\draw (9.975,-0.620)--(8.850,-1.240);
\draw (12.000,-0.620)--(12.000,-1.240);
\draw (9.975,-0.620)--(11.100,-1.240);
\draw (10.600,-1.240)--(10.600,-1.860);
\draw (10.100,-1.860)--(10.100,-2.480);
\draw (9.600,-2.480)--(9.600,-3.100);
\draw (12.000,-1.240)--(12.000,-1.860);
\draw (8.650,0.000)--(14.275,-0.620);
\draw (14.275,-0.620)--(14.150,-1.240);
\draw (14.150,-1.240)--(13.900,-1.860);
\draw (14.275,-0.620)--(14.900,-1.240);
\draw (14.150,-1.240)--(14.400,-1.860);
\draw (13.900,-1.860)--(13.900,-2.480);
\draw (8.650,0.000)--(16.300,-0.620);
\draw (14.275,-0.620)--(13.150,-1.240);
\draw (16.300,-0.620)--(16.300,-1.240);
\draw (14.275,-0.620)--(15.400,-1.240);
\draw (14.900,-1.240)--(14.900,-1.860);
\draw (14.400,-1.860)--(14.400,-2.480);
\draw (13.900,-2.480)--(13.900,-3.100);
\draw (16.300,-1.240)--(16.300,-1.860);
\draw (0.000,-1.240)--(0.000,-1.860);
\draw (4.300,-1.240)--(3.800,-1.860);
\draw (4.300,-1.240)--(4.300,-1.860);
\draw (8.850,-1.240)--(8.600,-1.860);
\draw (8.850,-1.240)--(9.100,-1.860);
\draw (13.150,-1.240)--(13.400,-1.860);
\draw (4.300,-1.240)--(4.800,-1.860);
\draw (13.150,-1.240)--(12.900,-1.860);
\node[fx] at (8.650,0.000) {$0$};
\node[fx] at (1.000,-0.620) {$S_{a}$};
\node[fx] at (0.750,-1.240) {$T_{a}$};
\node[cn] at (0.500,-1.860) {$U_{a}$};
\node[cn] at (1.500,-1.240) {$p_{a}$};
\node[cn] at (1.000,-1.860) {$q_{a}$};
\node[cn] at (0.500,-2.480) {$c^1_{a}$};
\node[cn] at (2.900,-0.620) {$c^{\mathrm{in}}_{a}$};
\node[cn] at (0.000,-1.240) {$c^{\mathrm{out}}_{a}$};
\node[cn] at (2.900,-1.240) {$w_{a}$};
\node[cn] at (2.000,-1.240) {$\bar w_{a}$};
\node[md] at (1.500,-1.860) {$y_{a}$};
\node[md] at (1.000,-2.480) {$\mu^1_{a}$};
\node[md] at (0.500,-3.100) {$\mu^2_{a}$};
\node[md] at (2.900,-1.860) {$\mu^0_{a}$};
\node[fx] at (5.550,-0.620) {$S_{b}$};
\node[fx] at (5.550,-1.240) {$T_{b}$};
\node[cn] at (5.300,-1.860) {$U_{b}$};
\node[cn] at (6.300,-1.240) {$p_{b}$};
\node[cn] at (5.800,-1.860) {$q_{b}$};
\node[cn] at (5.300,-2.480) {$c^1_{b}$};
\node[cn] at (7.700,-0.620) {$c^{\mathrm{in}}_{b}$};
\node[cn] at (4.300,-1.240) {$c^{\mathrm{out}}_{b}$};
\node[cn] at (7.700,-1.240) {$w_{b}$};
\node[cn] at (6.800,-1.240) {$\bar w_{b}$};
\node[md] at (6.300,-1.860) {$y_{b}$};
\node[md] at (5.800,-2.480) {$\mu^1_{b}$};
\node[md] at (5.300,-3.100) {$\mu^2_{b}$};
\node[md] at (7.700,-1.860) {$\mu^0_{b}$};
\node[fx] at (9.975,-0.620) {$S_{c}$};
\node[fx] at (9.850,-1.240) {$T_{c}$};
\node[cn] at (9.600,-1.860) {$U_{c}$};
\node[cn] at (10.600,-1.240) {$p_{c}$};
\node[cn] at (10.100,-1.860) {$q_{c}$};
\node[cn] at (9.600,-2.480) {$c^1_{c}$};
\node[cn] at (12.000,-0.620) {$c^{\mathrm{in}}_{c}$};
\node[cn] at (8.850,-1.240) {$c^{\mathrm{out}}_{c}$};
\node[cn] at (12.000,-1.240) {$w_{c}$};
\node[cn] at (11.100,-1.240) {$\bar w_{c}$};
\node[md] at (10.600,-1.860) {$y_{c}$};
\node[md] at (10.100,-2.480) {$\mu^1_{c}$};
\node[md] at (9.600,-3.100) {$\mu^2_{c}$};
\node[md] at (12.000,-1.860) {$\mu^0_{c}$};
\node[fx] at (14.275,-0.620) {$S_{d}$};
\node[fx] at (14.150,-1.240) {$T_{d}$};
\node[cn] at (13.900,-1.860) {$U_{d}$};
\node[cn] at (14.900,-1.240) {$p_{d}$};
\node[cn] at (14.400,-1.860) {$q_{d}$};
\node[cn] at (13.900,-2.480) {$c^1_{d}$};
\node[cn] at (16.300,-0.620) {$c^{\mathrm{in}}_{d}$};
\node[cn] at (13.150,-1.240) {$c^{\mathrm{out}}_{d}$};
\node[cn] at (16.300,-1.240) {$w_{d}$};
\node[cn] at (15.400,-1.240) {$\bar w_{d}$};
\node[md] at (14.900,-1.860) {$y_{d}$};
\node[md] at (14.400,-2.480) {$\mu^1_{d}$};
\node[md] at (13.900,-3.100) {$\mu^2_{d}$};
\node[md] at (16.300,-1.860) {$\mu^0_{d}$};
\node[md] at (0.000,-1.860) {$x_{ab}$};
\node[md] at (3.800,-1.860) {$x_{ba}$};
\node[md] at (4.300,-1.860) {$x_{bc}$};
\node[md] at (8.600,-1.860) {$x_{cb}$};
\node[md] at (9.100,-1.860) {$x_{cd}$};
\node[md] at (13.400,-1.860) {$x_{dc}$};
\node[md] at (4.800,-1.860) {$x_{bd}$};
\node[md] at (12.900,-1.860) {$x_{db}$};
\end{tikzpicture}}

\medskip
\resizebox{\linewidth}{!}{\begin{tikzpicture}[font=\scriptsize,fx/.style={circle,draw,fill=white,inner sep=0.2pt,minimum size=13pt},cn/.style={circle,draw,fill=red!20,inner sep=0.2pt,minimum size=13pt},md/.style={rectangle,draw,fill=blue!15,inner sep=1pt,minimum size=10pt}]
\node[font=\small\bfseries] at (8.650,0.6) {T};
\draw (8.650,0.000)--(1.000,-0.620);
\draw (1.000,-0.620)--(0.500,-1.240);
\draw (1.000,-0.620)--(1.000,-1.240);
\draw (0.500,-1.240)--(0.500,-1.860);
\draw (1.000,-1.240)--(1.000,-1.860);
\draw (1.000,-0.620)--(1.500,-1.240);
\draw (1.000,-0.620)--(0.000,-1.240);
\draw (8.650,0.000)--(2.900,-0.620);
\draw (1.000,-0.620)--(2.000,-1.240);
\draw (2.900,-0.620)--(2.900,-1.240);
\draw (1.000,-1.860)--(1.000,-2.480);
\draw (1.500,-1.240)--(1.500,-1.860);
\draw (2.900,-1.240)--(2.900,-1.860);
\draw (0.500,-1.860)--(0.500,-2.480);
\draw (8.650,0.000)--(5.550,-0.620);
\draw (5.550,-0.620)--(5.300,-1.240);
\draw (5.550,-0.620)--(5.800,-1.240);
\draw (5.300,-1.240)--(5.300,-1.860);
\draw (5.800,-1.240)--(5.800,-1.860);
\draw (5.550,-0.620)--(6.300,-1.240);
\draw (5.550,-0.620)--(4.300,-1.240);
\draw (8.650,0.000)--(7.700,-0.620);
\draw (5.550,-0.620)--(6.800,-1.240);
\draw (7.700,-0.620)--(7.700,-1.240);
\draw (5.800,-1.860)--(5.800,-2.480);
\draw (6.300,-1.240)--(6.300,-1.860);
\draw (7.700,-1.240)--(7.700,-1.860);
\draw (5.300,-1.860)--(5.300,-2.480);
\draw (8.650,0.000)--(9.975,-0.620);
\draw (9.975,-0.620)--(9.600,-1.240);
\draw (9.975,-0.620)--(10.100,-1.240);
\draw (9.600,-1.240)--(9.600,-1.860);
\draw (10.100,-1.240)--(10.100,-1.860);
\draw (9.975,-0.620)--(10.600,-1.240);
\draw (9.975,-0.620)--(8.850,-1.240);
\draw (8.650,0.000)--(12.000,-0.620);
\draw (9.975,-0.620)--(11.100,-1.240);
\draw (12.000,-0.620)--(12.000,-1.240);
\draw (10.100,-1.860)--(10.100,-2.480);
\draw (10.600,-1.240)--(10.600,-1.860);
\draw (12.000,-1.240)--(12.000,-1.860);
\draw (9.600,-1.860)--(9.600,-2.480);
\draw (8.650,0.000)--(14.275,-0.620);
\draw (14.275,-0.620)--(13.900,-1.240);
\draw (14.275,-0.620)--(14.400,-1.240);
\draw (13.900,-1.240)--(13.900,-1.860);
\draw (14.400,-1.240)--(14.400,-1.860);
\draw (14.275,-0.620)--(14.900,-1.240);
\draw (14.275,-0.620)--(13.150,-1.240);
\draw (8.650,0.000)--(16.300,-0.620);
\draw (14.275,-0.620)--(15.400,-1.240);
\draw (16.300,-0.620)--(16.300,-1.240);
\draw (14.400,-1.860)--(14.400,-2.480);
\draw (14.900,-1.240)--(14.900,-1.860);
\draw (16.300,-1.240)--(16.300,-1.860);
\draw (13.900,-1.860)--(13.900,-2.480);
\draw (4.300,-1.240)--(3.800,-1.860);
\draw (0.000,-1.240)--(0.000,-1.860);
\draw (8.850,-1.240)--(8.600,-1.860);
\draw (4.300,-1.240)--(4.300,-1.860);
\draw (13.150,-1.240)--(13.400,-1.860);
\draw (8.850,-1.240)--(9.100,-1.860);
\draw (13.150,-1.240)--(12.900,-1.860);
\draw (4.300,-1.240)--(4.800,-1.860);
\node[fx] at (8.650,0.000) {$0$};
\node[fx] at (1.000,-0.620) {$S_{a}$};
\node[fx] at (0.500,-1.240) {$T_{a}$};
\node[cn] at (1.000,-1.240) {$U_{a}$};
\node[cn] at (0.500,-1.860) {$p_{a}$};
\node[cn] at (1.000,-1.860) {$q_{a}$};
\node[cn] at (1.500,-1.240) {$c^1_{a}$};
\node[cn] at (0.000,-1.240) {$c^{\mathrm{in}}_{a}$};
\node[cn] at (2.900,-0.620) {$c^{\mathrm{out}}_{a}$};
\node[cn] at (2.000,-1.240) {$w_{a}$};
\node[cn] at (2.900,-1.240) {$\bar w_{a}$};
\node[md] at (1.000,-2.480) {$y_{a}$};
\node[md] at (1.500,-1.860) {$\mu^1_{a}$};
\node[md] at (2.900,-1.860) {$\mu^2_{a}$};
\node[md] at (0.500,-2.480) {$\mu^0_{a}$};
\node[fx] at (5.550,-0.620) {$S_{b}$};
\node[fx] at (5.300,-1.240) {$T_{b}$};
\node[cn] at (5.800,-1.240) {$U_{b}$};
\node[cn] at (5.300,-1.860) {$p_{b}$};
\node[cn] at (5.800,-1.860) {$q_{b}$};
\node[cn] at (6.300,-1.240) {$c^1_{b}$};
\node[cn] at (4.300,-1.240) {$c^{\mathrm{in}}_{b}$};
\node[cn] at (7.700,-0.620) {$c^{\mathrm{out}}_{b}$};
\node[cn] at (6.800,-1.240) {$w_{b}$};
\node[cn] at (7.700,-1.240) {$\bar w_{b}$};
\node[md] at (5.800,-2.480) {$y_{b}$};
\node[md] at (6.300,-1.860) {$\mu^1_{b}$};
\node[md] at (7.700,-1.860) {$\mu^2_{b}$};
\node[md] at (5.300,-2.480) {$\mu^0_{b}$};
\node[fx] at (9.975,-0.620) {$S_{c}$};
\node[fx] at (9.600,-1.240) {$T_{c}$};
\node[cn] at (10.100,-1.240) {$U_{c}$};
\node[cn] at (9.600,-1.860) {$p_{c}$};
\node[cn] at (10.100,-1.860) {$q_{c}$};
\node[cn] at (10.600,-1.240) {$c^1_{c}$};
\node[cn] at (8.850,-1.240) {$c^{\mathrm{in}}_{c}$};
\node[cn] at (12.000,-0.620) {$c^{\mathrm{out}}_{c}$};
\node[cn] at (11.100,-1.240) {$w_{c}$};
\node[cn] at (12.000,-1.240) {$\bar w_{c}$};
\node[md] at (10.100,-2.480) {$y_{c}$};
\node[md] at (10.600,-1.860) {$\mu^1_{c}$};
\node[md] at (12.000,-1.860) {$\mu^2_{c}$};
\node[md] at (9.600,-2.480) {$\mu^0_{c}$};
\node[fx] at (14.275,-0.620) {$S_{d}$};
\node[fx] at (13.900,-1.240) {$T_{d}$};
\node[cn] at (14.400,-1.240) {$U_{d}$};
\node[cn] at (13.900,-1.860) {$p_{d}$};
\node[cn] at (14.400,-1.860) {$q_{d}$};
\node[cn] at (14.900,-1.240) {$c^1_{d}$};
\node[cn] at (13.150,-1.240) {$c^{\mathrm{in}}_{d}$};
\node[cn] at (16.300,-0.620) {$c^{\mathrm{out}}_{d}$};
\node[cn] at (15.400,-1.240) {$w_{d}$};
\node[cn] at (16.300,-1.240) {$\bar w_{d}$};
\node[md] at (14.400,-2.480) {$y_{d}$};
\node[md] at (14.900,-1.860) {$\mu^1_{d}$};
\node[md] at (16.300,-1.860) {$\mu^2_{d}$};
\node[md] at (13.900,-2.480) {$\mu^0_{d}$};
\node[md] at (3.800,-1.860) {$x_{ab}$};
\node[md] at (0.000,-1.860) {$x_{ba}$};
\node[md] at (8.600,-1.860) {$x_{bc}$};
\node[md] at (4.300,-1.860) {$x_{cb}$};
\node[md] at (13.400,-1.860) {$x_{cd}$};
\node[md] at (9.100,-1.860) {$x_{dc}$};
\node[md] at (12.900,-1.860) {$x_{bd}$};
\node[md] at (4.800,-1.860) {$x_{db}$};
\end{tikzpicture}}
\caption{The instance of the reduction for the graph $G$ shown on top: the tree $S$ (middle) and the tree $T$ (bottom), with $n=64$ mutations and $K=80$. White: fixed; red: connectors; blue squares: mediators ($y_u=\langle p_u,q_u\rangle$, $\mu^0_u=\langle w_u,p_u\rangle$, $\mu^1_u=\langle q_u,c^1_u\rangle$, $\mu^2_u=\langle c^1_u,\bar w_u\rangle$, and the arc mediators $x_{uz}=\langle c^{\mathrm{out}}_u,c^{\mathrm{in}}_z\rangle$ for every edge $\{u,z\}$).}\label{fig:instance}
\end{figure}

\begin{table}[t]
\centering\scriptsize
\caption{An optimal sequence of $K+\tau(G)=82$ elementary moves for the instance of \autoref{fig:instance}, built as in \autoref{lem:np-ub} from the cover $\{b,c\}$. Each entry gives the moved mutation and its old and new parent. Moves 1--6 are the cuts.}\label{tab:solution}
\begin{tabular}{rll@{\hspace{1.0em}}rll@{\hspace{1.0em}}rll}
\toprule
\# & mutation & move & \# & mutation & move & \# & mutation & move\\
\midrule
1 & $y_{b}$ & $p_{b}\to S_{b}$ & 29 & $x_{bd}$ & $0\to c^{\mathrm{in}}_{d}$ & 57 & $c^1_{a}$ & $U_{a}\to S_{a}$\\
2 & $y_{b}$ & $S_{b}\to T_{b}$ & 30 & $c^{\mathrm{out}}_{c}$ & $S_{c}\to 0$ & 58 & $\mu^2_{a}$ & $c^1_{a}\to S_{a}$\\
3 & $y_{b}$ & $T_{b}\to q_{b}$ & 31 & $x_{cb}$ & $c^{\mathrm{out}}_{c}\to 0$ & 59 & $\mu^2_{a}$ & $S_{a}\to \bar w_{a}$\\
4 & $y_{c}$ & $p_{c}\to S_{c}$ & 32 & $x_{cb}$ & $0\to c^{\mathrm{in}}_{b}$ & 60 & $c^1_{d}$ & $U_{d}\to S_{d}$\\
5 & $y_{c}$ & $S_{c}\to T_{c}$ & 33 & $x_{cd}$ & $c^{\mathrm{out}}_{c}\to 0$ & 61 & $\mu^2_{d}$ & $c^1_{d}\to S_{d}$\\
6 & $y_{c}$ & $T_{c}\to q_{c}$ & 34 & $x_{cd}$ & $0\to c^{\mathrm{in}}_{d}$ & 62 & $\mu^2_{d}$ & $S_{d}\to \bar w_{d}$\\
7 & $q_{b}$ & $T_{b}\to U_{b}$ & 35 & $c^{\mathrm{in}}_{a}$ & $0\to S_{a}$ & 63 & $\bar w_{a}$ & $S_{a}\to c^{\mathrm{out}}_{a}$\\
8 & $\mu^1_{b}$ & $q_{b}\to U_{b}$ & 36 & $c^{\mathrm{in}}_{d}$ & $0\to S_{d}$ & 64 & $\bar w_{d}$ & $S_{d}\to c^{\mathrm{out}}_{d}$\\
9 & $\mu^1_{b}$ & $U_{b}\to c^1_{b}$ & 37 & $w_{a}$ & $c^{\mathrm{in}}_{a}\to S_{a}$ & 65 & $c^{\mathrm{out}}_{a}$ & $S_{a}\to 0$\\
10 & $q_{c}$ & $T_{c}\to U_{c}$ & 38 & $\mu^0_{a}$ & $w_{a}\to S_{a}$ & 66 & $x_{ab}$ & $c^{\mathrm{out}}_{a}\to 0$\\
11 & $\mu^1_{c}$ & $q_{c}\to U_{c}$ & 39 & $\mu^0_{a}$ & $S_{a}\to p_{a}$ & 67 & $x_{ab}$ & $0\to c^{\mathrm{in}}_{b}$\\
12 & $\mu^1_{c}$ & $U_{c}\to c^1_{c}$ & 40 & $w_{d}$ & $c^{\mathrm{in}}_{d}\to S_{d}$ & 68 & $c^{\mathrm{out}}_{d}$ & $S_{d}\to 0$\\
13 & $U_{b}$ & $T_{b}\to S_{b}$ & 41 & $\mu^0_{d}$ & $w_{d}\to S_{d}$ & 69 & $x_{dc}$ & $c^{\mathrm{out}}_{d}\to 0$\\
14 & $U_{c}$ & $T_{c}\to S_{c}$ & 42 & $\mu^0_{d}$ & $S_{d}\to p_{d}$ & 70 & $x_{dc}$ & $0\to c^{\mathrm{in}}_{c}$\\
15 & $c^1_{b}$ & $U_{b}\to S_{b}$ & 43 & $p_{a}$ & $S_{a}\to T_{a}$ & 71 & $x_{db}$ & $c^{\mathrm{out}}_{d}\to 0$\\
16 & $\mu^2_{b}$ & $c^1_{b}\to S_{b}$ & 44 & $y_{a}$ & $p_{a}\to T_{a}$ & 72 & $x_{db}$ & $0\to c^{\mathrm{in}}_{b}$\\
17 & $\mu^2_{b}$ & $S_{b}\to \bar w_{b}$ & 45 & $y_{a}$ & $T_{a}\to q_{a}$ & 73 & $c^{\mathrm{in}}_{c}$ & $0\to S_{c}$\\
18 & $c^1_{c}$ & $U_{c}\to S_{c}$ & 46 & $p_{d}$ & $S_{d}\to T_{d}$ & 74 & $c^{\mathrm{in}}_{b}$ & $0\to S_{b}$\\
19 & $\mu^2_{c}$ & $c^1_{c}\to S_{c}$ & 47 & $y_{d}$ & $p_{d}\to T_{d}$ & 75 & $w_{c}$ & $c^{\mathrm{in}}_{c}\to S_{c}$\\
20 & $\mu^2_{c}$ & $S_{c}\to \bar w_{c}$ & 48 & $y_{d}$ & $T_{d}\to q_{d}$ & 76 & $\mu^0_{c}$ & $w_{c}\to S_{c}$\\
21 & $\bar w_{b}$ & $S_{b}\to c^{\mathrm{out}}_{b}$ & 49 & $q_{a}$ & $T_{a}\to U_{a}$ & 77 & $\mu^0_{c}$ & $S_{c}\to p_{c}$\\
22 & $\bar w_{c}$ & $S_{c}\to c^{\mathrm{out}}_{c}$ & 50 & $\mu^1_{a}$ & $q_{a}\to U_{a}$ & 78 & $w_{b}$ & $c^{\mathrm{in}}_{b}\to S_{b}$\\
23 & $c^{\mathrm{out}}_{b}$ & $S_{b}\to 0$ & 51 & $\mu^1_{a}$ & $U_{a}\to c^1_{a}$ & 79 & $\mu^0_{b}$ & $w_{b}\to S_{b}$\\
24 & $x_{ba}$ & $c^{\mathrm{out}}_{b}\to 0$ & 52 & $q_{d}$ & $T_{d}\to U_{d}$ & 80 & $\mu^0_{b}$ & $S_{b}\to p_{b}$\\
25 & $x_{ba}$ & $0\to c^{\mathrm{in}}_{a}$ & 53 & $\mu^1_{d}$ & $q_{d}\to U_{d}$ & 81 & $p_{c}$ & $S_{c}\to T_{c}$\\
26 & $x_{bc}$ & $c^{\mathrm{out}}_{b}\to 0$ & 54 & $\mu^1_{d}$ & $U_{d}\to c^1_{d}$ & 82 & $p_{b}$ & $S_{b}\to T_{b}$\\
27 & $x_{bc}$ & $0\to c^{\mathrm{in}}_{c}$ & 55 & $U_{a}$ & $T_{a}\to S_{a}$ &  &  & \\
28 & $x_{bd}$ & $c^{\mathrm{out}}_{b}\to 0$ & 56 & $U_{d}$ & $T_{d}\to S_{d}$ &  &  & \\
\bottomrule
\end{tabular}
\end{table}

\subsection{The lower bound}\label{sec:np-lb}

\subparagraph{Idea of the lower bound.}
Fix any sequence from $S$ to $T$. Every move beyond the $K$ moves of the intended solution is an extra move (a \emph{payer}), and we charge it to the vertex whose gadget it happens in. We then show that a vertex that receives no charge, called \emph{quiet}, behaves exactly as designed:
\begin{itemize}
\item each of its connectors moves once, directly (rigidity);
\item nothing from outside enters its region except through the intended entry and exit (gate);
\item its mediators enforce the order of its chain (internal order).
\end{itemize}
If two adjacent vertices were both quiet, their two chains and the two arcs between them would form the cycle of \autoref{fig:gadget}, which is impossible. Hence the charged vertices form a vertex cover, and the number of extra moves is at least $\tau(G)$ (\autoref{fig:proofmap}). The delicate part is the charging rule itself, because one extra move may influence several gadgets. A move is charged to the gadget where it \emph{lands}, not to the gadget of the mutation that makes it, and a mediator gets one free stop, so that the intended solution costs nothing.

\begin{figure}[t]
\centering
\resizebox{\linewidth}{!}{\begin{tikzpicture}[font=\scriptsize,>=Stealth,
  b/.style={draw,rounded corners=2pt,align=center,inner sep=3pt,fill=gray!8}]
\node[b] (pay) at (0,0) {extra moves $=|P|$\\(payers, \autoref{sec:np-lb})};
\node[b] (quiet) at (3.3,0) {$u$ quiet:\\no payer charged to $u$};
\node[b] (rig) at (6.5,0) {rigidity\\(\autoref{lem:np-rigid})};
\node[b] (gate) at (9.3,0.8) {gate\\(\autoref{lem:np-gate})};
\node[b] (int) at (9.3,-0.8) {internal order\\(\autoref{lem:np-internal})};
\node[b] (edge) at (12.3,0.8) {edge order\\(\autoref{lem:np-edge})};
\node[b] (cyc) at (12.3,-0.8) {two quiet neighbours\\$\Rightarrow$ cycle: impossible};
\node[b] (vc) at (15.6,0) {charged vertices\\form a vertex cover};
\draw[->] (pay)--(quiet); \draw[->] (quiet)--(rig); \draw[->] (rig)--(gate); \draw[->] (rig)--(int);
\draw[->] (gate)--(edge); \draw[->] (edge)--(cyc); \draw[->] (int)--(cyc); \draw[->] (cyc)--(vc);
\end{tikzpicture}}
\caption{Structure of the proof of the lower bound $\de(S,T)\ge K+\tau(G)$.}\label{fig:proofmap}
\end{figure}

The lower bound is a charging argument. We first record the properties of the construction that the proof uses. All of them can be read off \autoref{tab:gadget}.

\begin{enumerate}
\item[(F1)] For every mediator $y=\langle p,q\rangle$: $p,q$ are connectors, $t(p)=s(q)$, $C(p)\cap C(q)=\{t(p)\}$, $q\notin C(z)$ for every $z\in C(p)$, and $p\notin C(z)$ for every $z\in C(q)$. In particular $q\notin C(p)$ and $p\notin C(q)$.
\item[(F2)] The pairs $\{s(y),t(y)\}$ are pairwise distinct. No mediator is a canonical parent of any mutation.
\item[(F3)] No fixed node or connector has an endpoint of an internal mediator as a canonical parent. Internal mediators have pairwise distinct sources and pairwise distinct targets, and no source or target of an internal mediator is a $c^{\mathrm{in}}$ or a $c^{\mathrm{out}}$.
\item[(F4)] The partners of a connector of gadget $u$ lie in gadget $u$, except the arc partners: $c^{\mathrm{out}}_u$ and $c^{\mathrm{in}}_c$ for every neighbour $c$ of $u$, and $c^{\mathrm{in}}_u$ and $c^{\mathrm{out}}_a$ for every neighbour $a$. The nodes $S_u,T_u,U_u,c^{\mathrm{in}}_u,c^{\mathrm{out}}_u$ have no partner in gadget $u$.
\end{enumerate}

\subparagraph{Stops and excess.}
Fix a sequence $\sigma$ of moves from $S$ to $T$. The \emph{trajectory} of a mutation $v$ is the sequence of its parents $s(v)=a_0,a_1,\dots,a_m=t(v)$, where $m=m(v)$ is the number of moves of $v$. The intermediate parents $a_1,\dots,a_{m-1}$ are the \emph{stops} of $v$, each with the time at which it begins. A stop of a connector $v$ at $s(v)$ or $t(v)$ is \emph{improper}. A stop of a mediator at its port is \emph{standard}. With
\[
e(v)=m(v)-[v\in M]-[v\text{ is a mediator}]
\]
we have $|\sigma|-K=\sum_v e(v)$. Throughout we use that a move of $v$ from parent $a$ to parent $b$ requires $b=\pa(a)$ (up-move) or $\pa(b)=a$ (down-move into the sibling $b$) at that moment.

A mediator $y=\langle p,q\rangle$ with $m(y)=1$ moves from $p$ to $q$ directly, so at that moment $p$ is a child of $q$ or $q$ is a child of $p$. By (F1) this is a stop of $p$ under $q$ or of $q$ under $p$, and by (F2) such a \emph{pair stop} can serve at most one mediator. Assign to each mediator with one move one such stop and call it \emph{matched}. The set $P$ of \emph{payers} consists of:
\begin{itemize}
\item the unmatched stops of the connectors;
\item for every fixed mutation that moves, its stops and one extra unit;
\item for every mediator with $m\ge2$, all its stops except one, the \emph{free} stop: its standard stop if it has one, and otherwise its first stop.
\end{itemize}
Counting term by term gives $\sum_v e(v)=|P|$: matched stops cancel the mediators with one move, a fixed mutation with $j$ stops has $e=j+1$, and a mediator with $m\ge2$ moves has $e=m-2$ and $m-1$ stops.

\subparagraph{Locating payers.}
Every non-mediator other than $0$ belongs to one gadget; write $\varphi(r)$ for its vertex. For an internal mediator of $u$ let $\varphi=u$, and for $x_{ac}$ let $\varphi=a$. At any moment define $\mathrm{Loc}(r)=\varphi(r)$ if $r$ is not a mediator, and $\mathrm{Loc}(r)=\mathrm{Loc}(\pa(r))$ if $r$ is a mediator; $\mathrm{Loc}(0)$ is undefined. A stop of $v$ under $r$ that begins at time $\theta$ is assigned to the vertex $\psi=\mathrm{Loc}(r)$, evaluated at $\theta$, or to $\psi=\varphi(v)$ if this is undefined. The extra unit of a fixed mutation $v$ is assigned to $\varphi(v)$.

\begin{lemma}\label{lem:np-cover}
The set $H$ of vertices to which some payer is assigned is a vertex cover of $G$.
\end{lemma}

Since $|\sigma|-K=|P|\ge|H|\ge\tau(G)$ for every $\sigma$, \autoref{lem:np-cover} and \autoref{lem:np-ub} give $\de(S,T)=K+\tau(G)$. This proves \autoref{thm:np}.

To prove \autoref{lem:np-cover}, call a vertex $u$ \emph{quiet} if no payer is assigned to $u$. Directly from the definition of $\psi$, a quiet vertex $u$ satisfies:
\begin{itemize}
\item[(Q1)] $S_u$ and $T_u$ never move.
\item[(Q2)] No connector of $u$ has an improper stop, a stop under $0$, a stop under a mediator located in $u$, or a stop under a node of gadget $u$ other than a partner. Its stops under partners are matched.
\item[(Q3)] No fixed node or connector outside gadget $u$ has a stop under a non-mediator of gadget $u$ or under a mediator located in $u$, except the matched arc stops ($c^{\mathrm{out}}_a$ under $c^{\mathrm{in}}_u$, and $c^{\mathrm{in}}_c$ under $c^{\mathrm{out}}_u$).
\item[(Q4)] Every mediator has at most one stop assigned to $u$, and that stop is its free stop.
\end{itemize}
For (Q2), note that $s(v),t(v)$ lie in gadget $u$ or are $0$ for a connector $v$ of $u$, so improper stops are assigned to $u$.

\begin{lemma}[rigidity]\label{lem:np-rigid}
If $u$ is quiet, every connector of $u$ moves exactly once.
\end{lemma}

\begin{lemma}[gate]\label{lem:np-gate}
Let $u$ be quiet and let $R_u$ be the set of proper descendants of $S_u$. A mutation enters or leaves $R_u$ only together with the move of $c^{\mathrm{in}}_u$ (which enters $R_u$ with its subtree) or of $c^{\mathrm{out}}_u$ (which leaves it with its subtree).
\end{lemma}

\begin{lemma}[internal order]\label{lem:np-internal}
If $u$ is quiet, then for every internal mediator $\langle p,q\rangle$ of $u$ the (unique) move of $p$ precedes that of $q$.
\end{lemma}

\begin{lemma}[edges]\label{lem:np-edge}
If $a$ and $c$ are quiet and $\{a,c\}\in E$, then $c^{\mathrm{out}}_a$ moves before $c^{\mathrm{in}}_c$.
\end{lemma}
\begin{proof}
At the start $x_{ac}$ is in $R_a$, below $c^{\mathrm{out}}_a$, and at the end it is in $R_c$, below $c^{\mathrm{in}}_c$. The sets $R_a$ and $R_c$ are disjoint, because $S_a$ and $S_c$ are children of $0$ throughout (Q1). By \autoref{lem:np-gate}, $x_{ac}$ leaves $R_a$ only with the move of $c^{\mathrm{out}}_a$ and enters $R_c$ only with the move of $c^{\mathrm{in}}_c$. It must leave $R_a$ before it enters $R_c$.
\end{proof}

\begin{proof}[Proof of \autoref{lem:np-cover}]
Suppose $a$ and $c$ are quiet and adjacent. By \autoref{lem:np-rigid} every connector of $a$ and of $c$ moves once, directly, so the orders of \autoref{tab:gadget} hold. By \autoref{lem:np-internal} the internal mediators give $w_u<p_u<q_u<c^1_u<\bar w_u$ for $u\in\{a,c\}$. Hence
\[
c^{\mathrm{in}}_u<w_u<p_u<q_u<U_u<c^1_u<\bar w_u<c^{\mathrm{out}}_u\qquad(u\in\{a,c\}).
\]
By \autoref{lem:np-edge}, $c^{\mathrm{out}}_a<c^{\mathrm{in}}_c$ and $c^{\mathrm{out}}_c<c^{\mathrm{in}}_a$. Together these give $c^{\mathrm{in}}_a<c^{\mathrm{in}}_a$, a contradiction. So every edge has an endpoint in $H$.
\end{proof}

The proofs of Lemmas~\ref{lem:np-rigid}--\ref{lem:np-internal} are in \autoref{app:np}.

\begin{corollary}\label{cor:median}
Computing a median of a set of mutation trees under $\de$ is \NP-hard, already for two input trees.
\end{corollary}
\begin{proof}
The median value of $\{S,T\}$ is exactly $\de(S,T)$: one of the inputs attains it, and the triangle inequality gives the lower bound.
\end{proof}

\section{One cluster at a time: nni+ is polynomial and gNNI is NP-complete}\label{sec:setcomplexity}

The reduction of \autoref{sec:complexity} charges extra moves to individual mutations. For \nnip and gNNI the right unit is the cluster: every move changes exactly one cluster (\autoref{lem:effect}). In this section we turn this into a reconfiguration view of both distances, in which a tree is a family of pairwise compatible clusters and a move replaces one of them. This view gives exact formulas for both distances. For \nnip the formula is computable in polynomial time, and for gNNI it leads to \NP-completeness by a short reduction.

\subsection{Trees as families of compatible clusters}

\begin{definition}\label{def:compatible}
For distinct labels $v,x\in L$ and sets $A,B\subseteq L$ with $v\in A$ and $x\in B$, the pairs $(A,v)$ and $(B,x)$ are \emph{compatible} if $A\cap B=\emptyset$, or $B\subsetneq A$ and $v\notin B$, or $A\subsetneq B$ and $x\notin A$. A \emph{cluster family} is a tuple $(A_v)_{v\in L}$ with $v\in A_v\subseteq L$ for all $v$.
\end{definition}

Compatibility says that $A$ and $B$ can be the clusters of $v$ and $x$ in one tree: either the two subtrees are disjoint or one of the two labels lies strictly below the other (\autoref{fig:compat}). Two sets that cross cannot be clusters of one tree, and neither can two nested sets in which the owner of the larger set lies in the smaller one, since then each owner would lie below the other.

\begin{figure}[t]
\centering
\begin{tikzpicture}[font=\scriptsize,
  lab/.style={circle,fill=black,inner sep=0pt,minimum size=3pt},
  setA/.style={draw=red!70!black,thick,fill=red!6},
  setB/.style={draw=blue!70!black,thick,fill=blue!6,fill opacity=0.6,text opacity=1},
  tn/.style={circle,draw,inner sep=0.5pt,minimum size=10pt,font=\tiny}]
\begin{scope}
\node at (0.9,1.05) {\textbf{(a)} disjoint};
\draw[setA] (0.35,0) ellipse (0.45 and 0.6); \draw[setB] (1.45,0) ellipse (0.45 and 0.6);
\node[lab,label={[font=\scriptsize]left:$v$}] at (0.45,0.2) {}; \node[lab] at (0.3,-0.25) {};
\node[lab,label={[font=\scriptsize]right:$x$}] at (1.35,0.2) {}; \node[lab] at (1.5,-0.25) {};
\node[red!70!black] at (0.35,-0.8) {$A$}; \node[blue!70!black] at (1.45,-0.8) {$B$};
\node[tn] (r) at (0.9,-1.25) {0}; \node[tn] (a) at (0.5,-1.85) {$v$}; \node[tn] (b) at (1.3,-1.85) {$x$}; \draw (r)--(a) (r)--(b);
\node[green!50!black] at (0.9,-2.35) {compatible};
\end{scope}
\begin{scope}[xshift=3.3cm]
\node at (0.9,1.05) {\textbf{(b)} $B\subsetneq A$, $v\notin B$};
\draw[setA] (0.9,0) ellipse (0.95 and 0.65); \draw[setB] (1.2,-0.05) ellipse (0.5 and 0.42);
\node[lab,label={[font=\scriptsize]above:$v$}] at (0.35,0.05) {};
\node[lab,label={[font=\scriptsize]above:$x$}] at (1.05,-0.05) {}; \node[lab] at (1.35,-0.15) {};
\node[red!70!black] at (-0.22,0) {$A$}; \node[blue!70!black] at (2.02,-0.05) {$B$};
\node[tn] (r) at (0.9,-1.15) {0}; \node[tn] (a) at (0.9,-1.6) {$v$}; \node[tn] (b) at (0.9,-2.05) {$x$}; \draw (r)--(a)--(b);
\node[green!50!black] at (0.9,-2.45) {compatible};
\end{scope}
\begin{scope}[xshift=6.6cm]
\node at (0.9,1.05) {\textbf{(c)} crossing};
\draw[setA] (0.55,0) ellipse (0.6 and 0.6); \draw[setB] (1.25,0) ellipse (0.6 and 0.6);
\node[lab,label={[font=\scriptsize]left:$v$}] at (0.35,0.15) {};
\node[lab,label={[font=\scriptsize]right:$x$}] at (1.45,0.15) {};
\node[lab,label={[font=\scriptsize]below:$y$}] at (0.9,0.05) {};
\node[red!70!black] at (0.3,-0.8) {$A$}; \node[blue!70!black] at (1.5,-0.8) {$B$};
\node[align=center,text width=3cm] at (0.9,-1.55) {$y$ would lie below $v$ and below $x$, but neither of $v,x$ lies below the other};
\node[red!70!black] at (0.9,-2.35) {incompatible};
\end{scope}
\begin{scope}[xshift=9.9cm]
\node at (0.9,1.05) {\textbf{(d)} $B\subsetneq A$, $v\in B$};
\draw[setA] (0.9,0) ellipse (0.95 and 0.65); \draw[setB] (1.05,-0.05) ellipse (0.62 and 0.4);
\node[lab,label={[font=\scriptsize]above:$v$}] at (0.8,-0.08) {};
\node[lab,label={[font=\scriptsize]above:$x$}] at (1.3,-0.08) {};
\node[lab] at (0.25,0.05) {};
\node[red!70!black] at (-0.22,0) {$A$}; \node[blue!70!black] at (1.88,-0.05) {$B$};
\node[align=center,text width=3cm] at (0.9,-1.55) {$x$ would lie below $v$ (as $B\subsetneq A$) and $v$ below $x$ (as $v\in B$)};
\node[red!70!black] at (0.9,-2.35) {incompatible};
\end{scope}
\end{tikzpicture}
\caption{The four ways in which the pairs $(A,v)$ and $(B,x)$ can meet (dots are labels). In (a) and (b) the sets can be the clusters of $v$ and $x$ in one tree, shown below. In (c) and (d) they cannot; (d) includes the case $A=B$.}\label{fig:compat}
\end{figure}
 The singleton $\{v\}$ is compatible with every $(B,x)$: if $v\in B$ then $\{v\}\subsetneq B$ and $x\notin\{v\}$, and otherwise the sets are disjoint. This simple fact drives the results of this section.

\begin{lemma}\label{lem:family}
A cluster family $(A_v)_{v\in L}$ is the family of clusters $(\cl_T(v))_{v\in L}$ of a mutation tree $T$ if and only if its members are pairwise compatible. In that case $T$ is unique.
\end{lemma}
\begin{proof}
In a tree, the clusters of $v$ and $x$ are disjoint unless one label is a proper descendant of the other, say $x$ of $v$, and then $\cl_T(x)\subsetneq\cl_T(v)$ and $v\notin\cl_T(x)$. So the clusters of a tree are pairwise compatible. Conversely, let the members be pairwise compatible. Two members that share a label are strictly nested, so the members that contain $A_v$ form a chain under strict inclusion. Let $\pa(v)$ be the label $x$ whose set $A_x$ is the smallest member strictly containing $A_v$, or $0$ if there is none. Parents have strictly larger sets, so following parents always ends at $0$ and $\pa$ defines a mutation tree $T$, in which $x$ is a proper ancestor of $v$ exactly when $A_v\subsetneq A_x$. We show $\cl_T(v)=A_v$. If $y\in\cl_T(v)$, then $y=v$ or $A_y\subsetneq A_v$, so $y\in A_v$. If $y\in A_v$ and $y\neq v$, then $A_y$ and $A_v$ share $y$, and compatibility gives $A_y\subsetneq A_v$ (the other nesting would require $y\notin A_v$), so $y$ is a descendant of $v$. Finally, a tree is determined by its clusters, since $x$ is a proper ancestor of $y$ exactly when $\cl(y)\subsetneq\cl(x)$.
\end{proof}

\begin{lemma}\label{lem:onecluster}
Let $P\neq Q$ be mutation trees. Then $Q$ is obtained from $P$ by one gNNI move if and only if $\cl_P(w)\neq\cl_Q(w)$ for exactly one label $w$. Moreover, $Q$ is obtained from $P$ by one \nnip move if and only if, in addition, $\cl_Q(w)\subsetneq\cl_P(w)$ or $\cl_P(w)\subsetneq\cl_Q(w)$.
\end{lemma}
\begin{proof}
A gNNI move on $(u,w)$ changes only $\cl(w)$: the cluster of $u$ keeps all moved nodes, and the moved subtrees are intact. Since $J_u$ and $J_w$ are not both empty, $\cl(w)$ does change. An \nnip move is a gNNI move with $J_u=\emptyset$ (the cluster only loses) or $J_w=\emptyset$ (it only gains), see \autoref{prop:relations}.

Conversely, let $w$ be the only label with $A=\cl_P(w)\neq A'=\cl_Q(w)$, and let $u=\pa_P(w)$. We write $\cl(z)$ for the common cluster of every other label $z$. First, $\pa_Q(w)=u$. Indeed, every label $z$ with $A'\subsetneq\cl(z)$ contains $w$, so in $P$ the sets $\cl(z)$ and $A$ are strictly nested, and $A\subsetneq\cl(z)$ because the other nesting would require $w\notin\cl(z)$. The same argument with $P$ and $Q$ exchanged shows that $A$ and $A'$ are strictly contained in exactly the same clusters, and so $w$ has the same parent in $P$ and $Q$.

Next, let $x\in\ch_P(w)$. If $x\in A'$, compatibility in $Q$ gives $\cl(x)\subseteq A'$ (because $w\notin\cl(x)$). If $x\notin A'$, then $\cl(x)\cap A'=\emptyset$, since neither nesting is possible. Likewise, every $y\in\ch_P(u)\setminus\{w\}$ satisfies $\cl(y)\subseteq A'$ if $y\in A'$ and $\cl(y)\cap A'=\emptyset$ otherwise. If $u\ne0$, then $A'\subsetneq\cl(u)$. In every case each label of $A'\setminus\{w\}$ lies in the cluster of a child of $w$ or of a sibling of $w$ in $P$. Let $J_w=\{x\in\ch_P(w):x\notin A'\}$ and $J_u=\{y\in\ch_P(u)\setminus\{w\}:y\in A'\}$. Then
\[
A'=\{w\}\cup\bigcup_{x\in\ch_P(w)\setminus J_w}\cl(x)\ \cup\ \bigcup_{y\in J_u}\cl(y),
\]
which is the cluster of $w$ after the gNNI move on $(u,w)$ that moves $J_u$ into $w$ and $J_w$ up to $u$. This move changes no other cluster, so by \autoref{lem:family} it produces $Q$. If $A'\subsetneq A$, then $J_u=\emptyset$, and the move is the up-move of $J_w$ from $w$. If $A\subsetneq A'$, then $J_w=\emptyset$, and it is the down-move of $J_u$ into $w$.
\end{proof}

So $\dgnni(S,T)$ is the minimum number of steps that turn the clusters of $S$ into those of $T$ when every step replaces one cluster by an arbitrary set and keeps the family pairwise compatible. $\dn(S,T)$ is the same with the restriction that a step only removes or only adds labels.

\subparagraph{Conflicts.}
Let $X=\{v\in L:\cl_S(v)\neq\cl_T(v)\}$, so $|X|=\dcl(S,T)$. The \emph{conflict digraph} $D(S,T)$ has vertex set $X$ and an arc $v\to x$ whenever $(\cl_S(v),v)$ and $(\cl_T(x),x)$ are incompatible. Consider a sequence of moves from $S$ to $T$. A label whose cluster changes exactly once goes from $\cl_S$ to $\cl_T$ in a single step. An arc $v\to x$ says that $x$ cannot hold its final cluster while $v$ still holds its initial one, because no tree has both (\autoref{fig:conflict}). If both labels change once, $v$ must therefore change first. One-way arcs only fix an order, but a directed cycle cannot be scheduled this way, and one of its labels has to change twice.

\begin{figure}[t]
\centering
\begin{tikzpicture}[font=\scriptsize,
  tn/.style={circle,draw,inner sep=0.5pt,minimum size=12pt},
  tl/.style={tn,fill=red!20},tg/.style={tn,fill=green!20},
  cl/.style={font=\tiny,text=gray!40!black,anchor=west},
  arc/.style={-{Stealth[length=5pt]},thick}]
\begin{scope}
\node at (0.9,0.55) {\textbf{(i)} $S$};
\node[tn] (r) at (0.9,0) {$0$};
\node[tl] (a) at (0.5,-0.7) {$a$}; \node[tg] (b) at (1.5,-0.7) {$b$};
\node[tg] (c) at (0.5,-1.4) {$c$}; \node[tn] (l) at (0.5,-2.1) {$\ell$};
\draw (r)--(a)--(c)--(l) (r)--(b);
\node[cl,anchor=east] at (a.west) {$\{a,c,\ell\}$}; \node[cl,anchor=east] at (c.west) {$\{c,\ell\}$}; \node[cl] at (b.east) {$\{b\}$};
\end{scope}
\begin{scope}[xshift=3.2cm]
\node at (0.4,0.55) {\textbf{(ii)} $T$};
\node[tn] (r) at (0.4,0) {$0$};
\node[tg] (c) at (0.4,-0.6) {$c$}; \node[tg] (b) at (0.4,-1.2) {$b$};
\node[tl] (a) at (0.4,-1.8) {$a$}; \node[tn] (l) at (0.4,-2.4) {$\ell$};
\draw (r)--(c)--(b)--(a)--(l);
\node[cl] at (c.east) {$\{a,b,c,\ell\}$}; \node[cl] at (b.east) {$\{a,b,\ell\}$}; \node[cl] at (a.east) {$\{a,\ell\}$};
\end{scope}
\begin{scope}[xshift=6.2cm]
\node at (1.0,0.55) {\textbf{(iii)} $D(S,T)$};
\node[tl] (a) at (0,-0.6) {$a$}; \node[tg] (c) at (2,-0.6) {$c$}; \node[tg] (b) at (1,-2.0) {$b$};
\draw[arc,red!75!black] (a) to[bend left=18] (c); \draw[arc,red!75!black] (c) to[bend left=18] (a);
\draw[arc,gray!60!black] (a)--(b); \draw[arc,gray!60!black] (c)--(b);
\end{scope}
\begin{scope}[xshift=9.4cm]
\node[align=left,anchor=north west,text width=3.5cm] at (0,0.3) {
$a\to b$: $\{a,c,\ell\}$ and $\{a,b,\ell\}$ cross, so $b$ waits for $a$.\\[2pt]
$c\to b$: $\{c,\ell\}$ and $\{a,b,\ell\}$ cross.\\[2pt]
$a\leftrightarrows c$: the tangle of \autoref{fig:tangle}.\\[2pt]
$b$ has no out-arc: $\{b\}$ is a singleton.};
\end{scope}
\end{tikzpicture}
\caption{A conflict digraph. Clusters are written next to the labels; $a$ only loses (red), $b$ and $c$ only gain (green). An arc $v\to x$ means that $x$ cannot hold its cluster in $T$ while $v$ still holds its cluster in $S$, so if both change once, $v$ changes first. The one-way arcs only fix an order. The 2-cycle forces $a$ or $c$ to change twice, so $\dn=\mathrm N+1=4$ and $\dgnni=\dcl+1=4$. An optimal sequence: $c$ loses $\ell$, $a$ loses $c$, $b$ gains $a$, $c$ gains $b$.}\label{fig:conflict}
\end{figure}

\begin{lemma}\label{lem:conflict}
If the clusters of $v$ and $x$ each change exactly once in a sequence from $S$ to $T$ and $v\to x$ is an arc of $D(S,T)$, then $v$ changes before $x$.
\end{lemma}
\begin{proof}
Otherwise some intermediate tree has $\cl(v)=\cl_S(v)$ and $\cl(x)=\cl_T(x)$, contradicting \autoref{lem:family}.
\end{proof}

A set $Y\subseteq X$ is a \emph{feedback vertex set} of a digraph if removing $Y$ leaves no directed cycle, and $\mathrm{fvs}(D)$ is the minimum size of such a set.

\subsection{gNNI: an exact formula and NP-completeness}

\begin{theorem}\label{thm:gnni}
For all $S,T\in\T_n$, $\dgnni(S,T)=\dcl(S,T)+\mathrm{fvs}(D(S,T))$.
\end{theorem}
\begin{proof}
($\ge$) Take an optimal gNNI sequence. Every label of $X$ changes its cluster at least once. Let $Q\subseteq X$ be the labels that change exactly once. By \autoref{lem:conflict} every arc of $D[Q]$ goes from an earlier to a later change, so $D[Q]$ is acyclic and $X\setminus Q$ is a feedback vertex set. Each label of $X\setminus Q$ changes at least twice, so the sequence has at least $|X|+|X\setminus Q|\ge\dcl+\mathrm{fvs}(D)$ moves.

($\le$) Let $Y$ be a minimum feedback vertex set and fix a topological order of the acyclic digraph $D-Y$. By \autoref{lem:onecluster}, it is enough to give a sequence of single-cluster replacements in which all intermediate families are pairwise compatible.
\begin{enumerate}
\item For each $y\in Y$, replace $\cl(y)$ by $\{y\}$.
\item For each $x\in X\setminus Y$ in topological order, replace $\cl_S(x)$ by $\cl_T(x)$.
\item For each $y\in Y$, replace $\{y\}$ by $\cl_T(y)$.
\end{enumerate}
Singletons are compatible with everything, so steps 1 and 3 keep compatibility. In step 3 all other clusters are singletons or target clusters. When $x$ changes in step 2, every other label holds its target cluster, a singleton, or its source cluster. In the last case, the label is either outside $X$, so its source cluster equals its target cluster and is compatible with $\cl_T(x)$, or it is a label $v\in X\setminus Y$ that comes later in the order. Then $v\to x$ is not an arc, so $(\cl_S(v),v)$ and $(\cl_T(x),x)$ are compatible. Every label of $X\setminus Y$ changes once and every label of $Y$ at most twice, which gives $\dcl+|Y|$ moves.
\end{proof}

The proof also shows what an extra move buys. A label that changes twice can pass through the singleton $\{y\}$, where it does not constrain anything, and this breaks every conflict cycle through $y$.

\begin{theorem}\label{thm:gnni-np}
\textsc{gNNI Distance} (rotation-free; given $S,T$ and $k$, is $\dgnni(S,T)\le k$?) is \NP-complete, even when $S$ and $T$ have depth two. Computing $\dgnni$ is \textsf{APX}-hard.
\end{theorem}
\begin{proof}
Membership in \NP\ follows from $\dgnni\le2n-2$. We reduce from \textsc{Vertex Cover}. Let $G=(V,E)$ be a graph without isolated vertices. The label set consists of one label $v$ for each vertex and two labels $\ell_{ac}$ and $\ell_{ca}$ for each edge $\{a,c\}\in E$. In both trees every vertex label is a child of the root. In $S$, the leaf $\ell_{ac}$ is a child of $a$, and in $T$ it is a child of $c$ (\autoref{fig:gnni}). Then $\cl_S(v)=\{v\}\cup\{\ell_{vc}:c\in N(v)\}$ and $\cl_T(v)=\{v\}\cup\{\ell_{cv}:c\in N(v)\}$, and the leaves have the same singleton cluster in both trees. So $X=V$ and $\dcl=|V|$. For $a\neq c$, the sets $\cl_S(a)$ and $\cl_T(c)$ contain neither $c$ nor $a$, respectively, so they are compatible exactly when they are disjoint, that is, when no leaf travels from $a$ to $c$. Hence $D(S,T)$ has the arcs $a\to c$ and $c\to a$ for every edge $\{a,c\}$ and no others. A feedback vertex set of this digraph must hit every 2-cycle, so it is a vertex cover, and every vertex cover leaves an independent set, which has no arcs. Therefore $\mathrm{fvs}(D)=\tau(G)$, and by \autoref{thm:gnni},
\[
\dgnni(S,T)=|V|+\tau(G).
\]
For \textsf{APX}-hardness, take $G$ cubic, where \textsc{Vertex Cover} is \textsf{APX}-hard~\cite{alimonti2000some}. Then $\tau(G)\ge|E|/3=|V|/2$, so $\dgnni\le3\tau(G)$. Conversely, from any sequence of length $k$ the labels of $V$ that change more than once form a vertex cover of size at most $k-|V|$ (proof of \autoref{thm:gnni}), and this cover is computed in linear time. This is an L-reduction~\cite{papadimitriou1991optimization}.
\end{proof}

\begin{figure}[t]
\centering
\begin{tikzpicture}[font=\scriptsize,every node/.style={circle,draw,inner sep=0.5pt,minimum size=12pt},
  lf/.style={rectangle,draw,fill=blue!12,inner sep=1.5pt,minimum size=9pt},
  vx/.style={fill=red!20}]
\begin{scope}
\node[draw=none] at (0.5,1.0) {\textbf{(i)} $G$};
\node[vx] (a) at (0,0) {$a$}; \node[vx] (b) at (1,0) {$b$}; \node[vx] (c) at (0.5,-0.85) {$c$};
\draw (a)--(b)--(c)--(a);
\end{scope}
\begin{scope}[xshift=3.3cm]
\node[draw=none] at (1.6,1.0) {\textbf{(ii)} $S$};
\node (r) at (1.6,0.4) {$0$};
\node[vx] (a) at (0.3,-0.4) {$a$}; \node[vx] (b) at (1.6,-0.4) {$b$}; \node[vx] (c) at (2.9,-0.4) {$c$};
\node[lf] (ab) at (0,-1.2) {$\ell_{ab}$}; \node[lf] (ac) at (0.6,-1.2) {$\ell_{ac}$};
\node[lf] (ba) at (1.3,-1.2) {$\ell_{ba}$}; \node[lf] (bc) at (1.9,-1.2) {$\ell_{bc}$};
\node[lf] (ca) at (2.6,-1.2) {$\ell_{ca}$}; \node[lf] (cb) at (3.2,-1.2) {$\ell_{cb}$};
\draw (r)--(a) (r)--(b) (r)--(c) (a)--(ab) (a)--(ac) (b)--(ba) (b)--(bc) (c)--(ca) (c)--(cb);
\end{scope}
\begin{scope}[xshift=7.6cm]
\node[draw=none] at (1.6,1.0) {\textbf{(iii)} $T$};
\node (r) at (1.6,0.4) {$0$};
\node[vx] (a) at (0.3,-0.4) {$a$}; \node[vx] (b) at (1.6,-0.4) {$b$}; \node[vx] (c) at (2.9,-0.4) {$c$};
\node[lf] (ba) at (0,-1.2) {$\ell_{ba}$}; \node[lf] (ca) at (0.6,-1.2) {$\ell_{ca}$};
\node[lf] (ab) at (1.3,-1.2) {$\ell_{ab}$}; \node[lf] (cb) at (1.9,-1.2) {$\ell_{cb}$};
\node[lf] (ac) at (2.6,-1.2) {$\ell_{ac}$}; \node[lf] (bc) at (3.2,-1.2) {$\ell_{bc}$};
\draw (r)--(a) (r)--(b) (r)--(c) (a)--(ba) (a)--(ca) (b)--(ab) (b)--(cb) (c)--(ac) (c)--(bc);
\end{scope}
\begin{scope}[xshift=12.1cm]
\node[draw=none] at (0.5,1.0) {\textbf{(iv)} $D(S,T)$};
\node[vx] (a) at (0,0) {$a$}; \node[vx] (b) at (1,0) {$b$}; \node[vx] (c) at (0.5,-0.85) {$c$};
\draw[-{Stealth[length=4pt]},bend left=15] (a) to (b); \draw[-{Stealth[length=4pt]},bend left=15] (b) to (a);
\draw[-{Stealth[length=4pt]},bend left=15] (b) to (c); \draw[-{Stealth[length=4pt]},bend left=15] (c) to (b);
\draw[-{Stealth[length=4pt]},bend left=15] (c) to (a); \draw[-{Stealth[length=4pt]},bend left=15] (a) to (c);
\end{scope}
\end{tikzpicture}

\vspace{2mm}
\begin{tabular}{rllll}
\toprule
\# & edge & moved into $w$ ($J_u$) & moved up to $u$ ($J_w$) & new cluster of $w$\\
\midrule
1 & $(0,a)$ & --- & $\ell_{ab},\ell_{ac}$ & $\{a\}$\\
2 & $(0,b)$ & --- & $\ell_{ba},\ell_{bc}$ & $\{b\}$\\
3 & $(0,c)$ & $\ell_{ac},\ell_{bc}$ & $\ell_{ca},\ell_{cb}$ & $\{c,\ell_{ac},\ell_{bc}\}$\\
4 & $(0,a)$ & $\ell_{ba},\ell_{ca}$ & --- & $\{a,\ell_{ba},\ell_{ca}\}$\\
5 & $(0,b)$ & $\ell_{ab},\ell_{cb}$ & --- & $\{b,\ell_{ab},\ell_{cb}\}$\\
\bottomrule
\end{tabular}
\caption{The reduction of \autoref{thm:gnni-np} for the triangle $G$. Every leaf $\ell_{xy}$ moves from $x$ in $S$ to $y$ in $T$, so the conflict digraph is $G$ with both orientations of every edge, and $\dgnni(S,T)=|V|+\tau(G)=5$. The table gives an optimal sequence built from the cover $\{a,b\}$ as in the proof of \autoref{thm:gnni}. Moves~1 and~2 empty $a$ and $b$, move~3 is the single gNNI move of $c$, which releases and collects leaves at the same time, and moves~4 and~5 fill $a$ and $b$. Without a second move on some vertex of each edge, the two leaves of that edge would have to change sides at the same time, which a single gNNI move on one endpoint cannot do.}\label{fig:gnni}
\end{figure}

\begin{corollary}\label{cor:gnni-fpt}
$\dgnni(S,T)$ is fixed-parameter tractable with respect to $k=\dgnni(S,T)-\dcl(S,T)$, by the algorithm for directed feedback vertex set of Chen et al.~\cite{chen2008fixed} applied to $D(S,T)$, which has at most $n$ vertices and is built in $O(n^3)$ time. Together with \autoref{thm:approx}, $\dgnni$ is \NP-hard, \textsf{APX}-hard, and approximable within a factor two. As in \autoref{cor:median}, computing a median under $\dgnni$ is \NP-hard already for two input trees.
\end{corollary}

\subparagraph{Rotations.}
The rooted gNNI of Jahn et al.~\cite{jahn2021bourque} also allows a rotation, in which $w$ takes the position of $u$ and $u$ becomes a child of $w$. Equivalently, it is their unrooted gNNI applied to the mutation tree regarded as an unrooted tree on $L\cup\{0\}$. A rotation hands the cluster of $u$ over to $w$ and gives $u$ a new cluster, so it changes the clusters of two labels, and \autoref{lem:onecluster} and \autoref{thm:gnni} do not apply to this variant. Our reduction does not carry over either: for the path $a$--$b$--$c$--$d$, rotations reduce the distance from $|V|+\tau(G)=6$ to $5$. The complexity of gNNI with rotations remains open.

\subsection{nni+: an exact formula in polynomial time}

Under \nnip a single step only removes labels from a cluster or only adds them. Let
\[
X_{-}=\{v:\cl_T(v)\subsetneq\cl_S(v)\},\qquad X_{+}=\{v:\cl_S(v)\subsetneq\cl_T(v)\},\qquad X_{\pm}=X\setminus(X_-\cup X_+),
\]
be the labels whose cluster only has to lose, only has to gain, and has to do both. Then $\mathrm N=|X_-|+|X_+|+2|X_\pm|$ (\autoref{sec:basic}). A label of $X_\pm$ needs two steps anyway, and it can take them as a loss to $\{v\}$ followed by a gain. In between it conflicts with nothing. This is precisely the ``parking'' that made the reduction of \autoref{sec:complexity} fail for \nnip: the instance of \autoref{fig:instance} has no tangle in the sense below, and $\dn=\mathrm N=36$. The only conflicts that can cost an extra \nnip move are therefore between labels that change once, and they have a very restricted form.

\begin{definition}\label{def:tangle}
A \emph{tangle} is a pair $(a,c)$ with $a\in X_-$ and $c\in X_+$ such that $(\cl_S(c),c)$ and $(\cl_T(a),a)$ are incompatible. The \emph{tangle graph} $B(S,T)$ is the bipartite graph on $X_-\cup X_+$ whose edges are the tangles, and $\nu(B)$ denotes the size of a maximum matching.
\end{definition}

By \autoref{lem:tangle} below, $(a,c)$ is a tangle exactly when $c$ lies below $a$ in $S$ and the final cluster of $a$ keeps some, but not all, of the clade of $c$. Let $\ell\in\cl_S(c)\cap\cl_T(a)$ be a label that $a$ keeps. The label $c$ itself is one that $a$ drops: otherwise $c$ would lie below $a$ in $T$, and $a$ would keep all of $\cl_T(c)\supseteq\cl_S(c)$ (\autoref{fig:tangle}). Then $a$ and $c$ wait for each other. Since $a$ only loses, it has to drop $c$ while keeping $\ell$, which is impossible as long as $\ell$ lies below $c$. If $c$ changes once, it only gains, so $\ell$ stays below $c$ until $c$ has reached its final cluster. But $c$ cannot take its final cluster while $a$ still holds its initial one, since these two are incompatible as well (\autoref{lem:tangle}(3)). So one of the two labels changes twice. Either $a$ drops its whole clade, $\ell$ included, and recovers $\ell$ later, or $c$ releases $\ell$ and gains everything later. Both are a parking, and both cost one extra move.

\begin{lemma}\label{lem:tangle}
$(a,c)$ with $a\in X_-$ and $c\in X_+$ is a tangle if and only if $c$ is a proper descendant of $a$ in $S$ and $\cl_T(a)$ contains some, but not all, labels of $\cl_S(c)$. Moreover, let $D_1=D(S,T)[X_-\cup X_+]$. Then:
\begin{enumerate}
\item every arc of $D_1$ between two labels of $X_-$ goes from $a$ to $b$ with $\cl_S(a)\subsetneq\cl_S(b)$;
\item every arc of $D_1$ between two labels of $X_+$ goes from $c$ to $d$ with $\cl_T(d)\subsetneq\cl_T(c)$;
\item the arcs from $X_+$ to $X_-$ are exactly the pairs $c\to a$ for tangles $(a,c)$, and every tangle also gives the arc $a\to c$.
\end{enumerate}
\end{lemma}
\begin{proof}
(1) Let $a,b\in X_-$ with $a\to b$. The sets $\cl_S(a)$ and $\cl_T(b)$ meet, and $\cl_T(b)\subseteq\cl_S(b)$, so $\cl_S(a)$ and $\cl_S(b)$ are strictly nested. If $\cl_S(b)\subsetneq\cl_S(a)$ with $a\notin\cl_S(b)$, then $\cl_T(b)\subsetneq\cl_S(a)$ and $a\notin\cl_T(b)$, so the pair would be compatible. Hence $\cl_S(a)\subsetneq\cl_S(b)$. (2) is (1) with the roles of $S$ and $T$ exchanged.

(3) Let $a\in X_-$, $c\in X_+$ and let $(\cl_S(c),c)$ and $(\cl_T(a),a)$ be incompatible. These sets meet, so $\cl_S(c)$ meets $\cl_S(a)\supseteq\cl_T(a)$, and the two are strictly nested in $S$. If $a$ were below $c$, then $\cl_T(a)\subseteq\cl_S(a)\subsetneq\cl_S(c)$ with $c\notin\cl_T(a)$, which is compatible. So $c$ is a proper descendant of $a$ in $S$, and $a\notin\cl_S(c)$. Then compatibility of $(\cl_S(c),c)$ and $(\cl_T(a),a)$ can only hold by disjointness or by $\cl_S(c)\subseteq\cl_T(a)$, since $\cl_T(a)\subsetneq\cl_S(c)$ is impossible because $a\notin\cl_S(c)$. This proves the characterization. For the arc $a\to c$, assume that $(\cl_S(a),a)$ and $(\cl_T(c),c)$ are compatible. Since $c\in\cl_S(a)$, this forces $\cl_T(c)\subsetneq\cl_S(a)$ and $a\notin\cl_T(c)$. In $T$, the sets $\cl_T(c)\supseteq\cl_S(c)$ and $\cl_T(a)$ meet and $a\notin\cl_T(c)$, so $\cl_T(c)\subsetneq\cl_T(a)$. Then $\cl_S(c)\subseteq\cl_T(a)$, contradicting the characterization.
\end{proof}

\begin{theorem}\label{thm:nniplus}
For all $S,T\in\T_n$, $\dn(S,T)=\mathrm N(S,T)+\nu(B(S,T))$, that is, $\mathrm N$ plus the maximum number of pairwise disjoint tangles. It can be computed, together with an optimal sequence, in $O(n^3)$ time.
\end{theorem}
\begin{proof}
First, $\mathrm{fvs}(D_1)=\nu(B)$. By \autoref{lem:tangle}(3), every tangle is a 2-cycle of $D_1$, so every feedback vertex set of $D_1$ is a vertex cover of $B$. Conversely, if $Y$ is a vertex cover of $B$, then $D_1-Y$ has no arc from $X_+$ to $X_-$. A cycle of $D_1-Y$ would therefore lie inside $X_-$ or inside $X_+$, which is impossible by \autoref{lem:tangle}(1),(2). By K\H{o}nig's theorem the minimum vertex cover of the bipartite graph $B$ has size $\nu(B)$~\cite{schrijver2003combinatorial}.

($\ge$) In an \nnip sequence, a label of $X_\pm$ changes at least twice and a label of $X_-\cup X_+$ at least once. Let $Q$ be the labels of $X_-\cup X_+$ that change exactly once. By \autoref{lem:conflict}, $D_1[Q]$ is acyclic, so the sequence has at least $\mathrm N+|(X_-\cup X_+)\setminus Q|\ge\mathrm N+\mathrm{fvs}(D_1)$ moves.

($\le$) Let $Y$ be a minimum vertex cover of $B$, and fix a topological order of $D_1-Y$.
\begin{enumerate}
\item Each $v\in X_\pm\cup Y$ loses all its descendants by an up-move of all its children. This makes $\cl(v)=\{v\}$.
\item Each $x\in(X_-\cup X_+)\setminus Y$, in topological order, goes from $\cl_S(x)$ to $\cl_T(x)$ by one loss or one gain.
\item Each $v$ with $\cl(v)=\{v\}\neq\cl_T(v)$ gains $\cl_T(v)\setminus\{v\}$.
\end{enumerate}
Compatibility is preserved exactly as in the proof of \autoref{thm:gnni}, and every step is a single loss or gain, so it is an \nnip move by \autoref{lem:onecluster}. Labels of $X_\pm$ use two moves and labels of $Y$ at most two. The remaining labels of $X$ use one move each, for a total of at most $\mathrm N+|Y|$. (Step 1 is a real move for $v\in Y$, because a label in a tangle has a source cluster different from its singleton.)

The clusters are computed in $O(n^2)$ time. Each of the $O(n^2)$ compatibility tests takes $O(n)$ time. A maximum matching and a minimum vertex cover of $B$ take $O(n^{5/2})$ time~\cite{hopcroft1973n52}. The three steps take $O(n)$ moves of $O(n)$ time each.
\end{proof}

\begin{figure}[t]
\centering
\begin{tikzpicture}[font=\scriptsize,every node/.style={circle,draw,inner sep=0.5pt,minimum size=12pt},
  tl/.style={fill=red!20},tg/.style={fill=green!20}]
\foreach \x/\lab/\caps in {0/{(i) $S$}/a,2.6/{(ii)}/b,5.2/{(iii)}/c,7.8/{(iv) $T$}/d}{\node[draw=none] at (\x+0.4,0.55) {\textbf{\lab}};}
\begin{scope}
\node (r) at (0.4,0) {$0$}; \node[tl] (a) at (0.4,-0.75) {$a$}; \node[tg] (c) at (0.4,-1.5) {$c$}; \node (l) at (0.4,-2.25) {$\ell$};
\draw (r)--(a)--(c)--(l);
\draw[draw=green!50!black,densely dashed,rounded corners] (0.08,-1.2) rectangle (0.72,-2.55);
\node[draw=none,green!40!black,anchor=east,inner sep=1pt] at (0.06,-1.88) {$\cl_S(c)$};
\draw[red!70!black,densely dashed,-{Stealth[length=4pt]}] (a) to[bend left=55] (l);
\node[draw=none,red!70!black,anchor=east,inner sep=1pt] at (0.1,-0.75) {$\cl_T(a)=\{a,\ell\}$};
\end{scope}
\begin{scope}[xshift=2.6cm]
\node (r) at (0.4,0) {$0$}; \node[tl] (a) at (0.4,-0.75) {$a$}; \node[tg] (c) at (0,-1.5) {$c$}; \node (l) at (0.8,-1.5) {$\ell$};
\draw (r)--(a)--(c) (a)--(l);
\end{scope}
\begin{scope}[xshift=5.2cm]
\node (r) at (0.4,0) {$0$}; \node[tl] (a) at (0,-0.75) {$a$}; \node[tg] (c) at (0.8,-0.75) {$c$}; \node (l) at (0,-1.5) {$\ell$};
\draw (r)--(a)--(l) (r)--(c);
\end{scope}
\begin{scope}[xshift=7.8cm]
\node (r) at (0.4,0) {$0$}; \node[tg] (c) at (0.4,-0.75) {$c$}; \node[tl] (a) at (0.4,-1.5) {$a$}; \node (l) at (0.4,-2.25) {$\ell$};
\draw (r)--(c)--(a)--(l);
\end{scope}
\foreach \x/\t in {1.2/{$\ell$ up from $c$},3.8/{$c$ up from $a$},6.4/{$a$ down into $c$}}{
\draw[->,thick,draw=gray] (\x,-1.1)--(\x+0.9,-1.1); \node[draw=none,align=center,text width=1.3cm] at (\x+0.45,-1.55) {\t};}
\end{tikzpicture}
\caption{The smallest tangle. In $S$, the clade of $c$ (dashed box) contains $\ell$, which $a$ keeps (red arrow), and $c$, which $a$ drops. Mutation $a$ only loses ($\{a,c,\ell\}\to\{a,\ell\}$) and $c$ only gains ($\{c,\ell\}\to\{a,c,\ell\}$), so $\mathrm N=2$. However, $a$ must give up $c$ but keep $\ell$, which is impossible while $\ell$ is below $c$. Meanwhile, $c$ cannot gain $a$ while $c$ is still below $a$. One of the two must change twice, and $\dn=\mathrm N+1=3$. Here $c$ first loses $\ell$ (ii), then $a$ makes its single loss (iii), and $c$ gains (iv). The tangle is also a 2-cycle of the conflict digraph, so $\dgnni=\dcl+1=3$ as well.}\label{fig:tangle}
\end{figure}

The theorem settles the complexity of \nnip and shows that the simple bound $\mathrm N$ is exact whenever there is no tangle. Tangles require a mutation $c$ that only gains and lies, in $S$, below a mutation $a$ that only loses and keeps part of the clade of $c$. This is a local order inversion such as the one in \autoref{fig:tangle}. In particular, the depth bound $\ddep$ of \autoref{thm:lb} never exceeds $\mathrm N+\nu(B)$; for the two reversed chains of \autoref{thm:diameter}, the end labels form no tangle and $\dn=\mathrm N=2n-2$.

\subparagraph{Three operations, three answers.}
The three local operations differ only in what one move may do to a single cluster. In \enni it may remove or add one child clade, in \nnip any set of child clades but in one direction, and in gNNI it may remove and add at the same time. The elementary distance is \NP-complete (\autoref{thm:np}) because the clades must travel one at a time. The \nnip distance is polynomial because a cluster that has to change in both directions can be emptied first, and the remaining conflicts form a bipartite structure. The (rotation-free) gNNI distance is \NP-complete although each gNNI move is more powerful than an \nnip move. Since one gNNI move suffices for any cluster, its baseline $\dcl$ is lower, and the remaining conflicts between labels changing once can form arbitrary 2-cycles, which encode \textsc{Vertex Cover}.

\begin{openproblem}\label{op:list}
(i) Is the gNNI distance with rotations \NP-hard? Can $\dgnni$ be approximated within a factor below two in polynomial time, for instance by exploiting the structure of conflict digraphs? (ii) Is $\de$ fixed-parameter tractable with respect to the excess $\de-\dpc$, and does it admit a constant-factor approximation? (iii) Does the median problem become polynomial for \nnip, or is it \NP-hard as for the elementary operation (\autoref{cor:median})?
\end{openproblem}

\section{Discussion and conclusion}\label{sec:conclusion}

We proposed \nnip as the local rearrangement for mutation trees. An \nnip move takes several sibling clades one level up or down, which is the typical effect of a single misplaced mutation. It lies between its elementary version \enni, which relocates one clade at a time, and gNNI, which exchanges clades across an edge in both directions. Viewing mutation trees as families of compatible clusters, we showed that the three operations have different complexity. The \nnip distance equals $\mathrm N$ plus the maximum number of disjoint tangles and is computable in $O(n^3)$ time. The rotation-free gNNI distance equals $\dcl$ plus a minimum feedback vertex set of a conflict digraph, and it is \NP-complete and \textsf{APX}-hard. The elementary distance is \NP-complete, and it equals the parent--child distance exactly when no mutation is dead.

The biological reading is direct. One \nnip move declares that a mutation occurred before (or after) a group of subclones branched off, and a tangle is a pair of mutations whose order is inverted while one of them keeps part of the other's clade. Several extensions are natural. First, real mutation trees often place several mutations on one node (clones). The operation and the cluster view extend verbatim to set-labelled trees with a common partition, while comparing trees with different partitions or label sets needs ideas such as those behind the Bourque distances~\cite{jahn2021bourque}. Second, for bulk data, trees must additionally satisfy the sum condition on clone frequencies~\cite{elkebir2015reconstruction}, and restricting \nnip to feasible trees gives a constrained distance that deserves study. Third, the polynomial formula makes consensus and median questions for \nnip worth revisiting (\autoref{op:list}).

\bibliography{refs}

@article{jahn2016tree,
  title={Tree inference for single-cell data},
  author={Jahn, Katharina and Kuipers, Jack and Beerenwinkel, Niko},
  journal={Genome Biology},
  volume={17},
  pages={86},
  year={2016}
}

@article{jahn2021bourque,
  title={The {Bourque} distances for mutation trees of cancers},
  author={Jahn, Katharina and Beerenwinkel, Niko and Zhang, Louxin},
  journal={Algorithms for Molecular Biology},
  volume={16},
  pages={9},
  year={2021}
}

@inproceedings{dasgupta2000computing,
  title={On computing the nearest neighbor interchange distance},
  author={DasGupta, Bhaskar and He, Xin and Jiang, Tao and Li, Ming and Tromp, John and Zhang, Louxin},
  booktitle={Discrete Mathematical Problems with Medical Applications},
  series={DIMACS Series in Discrete Mathematics and Theoretical Computer Science},
  volume={55},
  pages={125--143},
  year={2000}
}

@article{li1996nni,
  title={On the nearest neighbour interchange distance between evolutionary trees},
  author={Li, Ming and Tromp, John and Zhang, Louxin},
  journal={Journal of Theoretical Biology},
  volume={182},
  number={4},
  pages={463--467},
  year={1996}
}

@article{robinson1981comparison,
  title={Comparison of phylogenetic trees},
  author={Robinson, David F. and Foulds, Leslie R.},
  journal={Mathematical Biosciences},
  volume={53},
  number={1--2},
  pages={131--147},
  year={1981}
}

@inproceedings{bernardini2019rearrangement,
  title={A rearrangement distance for fully-labelled trees},
  author={Bernardini, Giulia and Bonizzoni, Paola and {Della Vedova}, Gianluca and Patterson, Murray},
  booktitle={30th Annual Symposium on Combinatorial Pattern Matching (CPM 2019)},
  series={LIPIcs},
  volume={128},
  pages={28:1--28:15},
  year={2019}
}

@article{dinardo2020distance,
  title={Distance measures for tumor evolutionary trees},
  author={DiNardo, Zach and Tomlin, Kiran and Ritz, Anna and Oesper, Layla},
  journal={Bioinformatics},
  volume={36},
  number={7},
  pages={2090--2097},
  year={2020}
}

@article{ciccolella2021triplet,
  title={Triplet-based similarity score for fully multilabeled trees with poly-occurring labels},
  author={Ciccolella, Simone and Bernardini, Giulia and Denti, Luca and Bonizzoni, Paola and Previtali, Marco and {Della Vedova}, Gianluca},
  journal={Bioinformatics},
  volume={37},
  number={2},
  pages={178--184},
  year={2021}
}

@article{karpov2019mltd,
  title={A multi-labeled tree dissimilarity measure for comparing ``clonal trees'' of tumor progression},
  author={Karpov, Nikolai and Malikic, Salem and Rahman, Md. Khaledur and Sahinalp, S. Cenk},
  journal={Algorithms for Molecular Biology},
  volume={14},
  pages={17},
  year={2019}
}

@inproceedings{govek2018consensus,
  title={A consensus approach to infer tumor evolutionary histories},
  author={Govek, Kiya and Sikes, Camden and Oesper, Layla},
  booktitle={Proceedings of the 2018 ACM International Conference on Bioinformatics, Computational Biology, and Health Informatics},
  pages={63--72},
  year={2018}
}

@article{aguse2019summarizing,
  title={Summarizing the solution space in tumor phylogeny inference by multiple consensus trees},
  author={Aguse, Nuraini and Qi, Yuanyuan and El-Kebir, Mohammed},
  journal={Bioinformatics},
  volume={35},
  number={14},
  pages={i408--i416},
  year={2019}
}

@article{christensen2020recap,
  title={Detecting evolutionary patterns of cancers using consensus trees},
  author={Christensen, Sarah and Kim, Juho and Chia, Nicholas and Koyejo, Oluwasanmi and El-Kebir, Mohammed},
  journal={Bioinformatics},
  volume={36},
  number={Supplement\_2},
  pages={i684--i691},
  year={2020}
}

@article{elkebir2015reconstruction,
  title={Reconstruction of clonal trees and tumor composition from multi-sample sequencing data},
  author={El-Kebir, Mohammed and Oesper, Layla and Acheson-Field, Hannah and Raphael, Benjamin J.},
  journal={Bioinformatics},
  volume={31},
  number={12},
  pages={i62--i70},
  year={2015}
}

@article{gusfield1991efficient,
  title={Efficient algorithms for inferring evolutionary trees},
  author={Gusfield, Dan},
  journal={Networks},
  volume={21},
  number={1},
  pages={19--28},
  year={1991}
}

@article{kuipers2017advances,
  title={Advances in understanding tumour evolution through single-cell sequencing},
  author={Kuipers, Jack and Jahn, Katharina and Beerenwinkel, Niko},
  journal={Biochimica et Biophysica Acta (BBA) -- Reviews on Cancer},
  volume={1867},
  number={2},
  pages={127--138},
  year={2017}
}

@article{morita2020clonal,
  title={Clonal evolution of acute myeloid leukemia revealed by high-throughput single-cell genomics},
  author={Morita, Kiyomi and Wang, Feng and Jahn, Katharina and Hu, Tianyuan and Tanaka, Tomoyuki and Sasaki, Yuya and Kuipers, Jack and Loghavi, Sanam and Wang, Sa A. and Yan, Yuanqing and others},
  journal={Nature Communications},
  volume={11},
  pages={5327},
  year={2020}
}

@incollection{karp1972reducibility,
  author    = {Richard M. Karp},
  title     = {Reducibility among combinatorial problems},
  booktitle = {Complexity of Computer Computations},
  editor    = {Raymond E. Miller and James W. Thatcher},
  publisher = {Plenum Press},
  pages     = {85--103},
  year      = {1972},
  doi       = {10.1007/978-1-4684-2001-2_9}
}

@article{chen2008fixed,
  author  = {Jianer Chen and Yang Liu and Songjian Lu and Barry O'Sullivan and Igor Razgon},
  title   = {A fixed-parameter algorithm for the directed feedback vertex set problem},
  journal = {Journal of the ACM},
  volume  = {55},
  number  = {5},
  pages   = {21:1--21:19},
  year    = {2008},
  doi     = {10.1145/1411509.1411511}
}

@article{alimonti2000some,
  author  = {Paola Alimonti and Viggo Kann},
  title   = {Some {APX}-completeness results for cubic graphs},
  journal = {Theoretical Computer Science},
  volume  = {237},
  number  = {1--2},
  pages   = {123--134},
  year    = {2000},
  doi     = {10.1016/S0304-3975(98)00158-3}
}

@article{papadimitriou1991optimization,
  author  = {Christos H. Papadimitriou and Mihalis Yannakakis},
  title   = {Optimization, approximation, and complexity classes},
  journal = {Journal of Computer and System Sciences},
  volume  = {43},
  number  = {3},
  pages   = {425--440},
  year    = {1991},
  doi     = {10.1016/0022-0000(91)90023-X}
}

@article{hopcroft1973n52,
  author  = {John E. Hopcroft and Richard M. Karp},
  title   = {An $n^{5/2}$ algorithm for maximum matchings in bipartite graphs},
  journal = {SIAM Journal on Computing},
  volume  = {2},
  number  = {4},
  pages   = {225--231},
  year    = {1973},
  doi     = {10.1137/0202019}
}

@book{schrijver2003combinatorial,
  author    = {Alexander Schrijver},
  title     = {Combinatorial Optimization: Polyhedra and Efficiency},
  publisher = {Springer},
  year      = {2003}
}

\appendix
\section{Proofs for Section~\ref{sec:np-lb}}\label{app:np}

Throughout, $u$ is a quiet vertex and we use (F1)--(F4) and (Q1)--(Q4). Recall that a move of $v$ from parent $a$ to parent $b$ requires, at that moment, $b=\pa(a)$ (an up-move) or $\pa(b)=a$ (a down-move into the sibling $b$). By (Q2), a connector $v$ of $u$ is at $s(v)$ only before its first move and at $t(v)$ only after its last move. We first describe the children of the nodes of gadget $u$.

\begin{enumerate}
\item[(O)] Let $h$ be a non-mediator of gadget $u$ other than $c^{\mathrm{in}}_u,c^{\mathrm{out}}_u$. At every moment, each child of $h$ is a non-mediator of gadget $u$ or a mediator. A connector $v$ of $u$ can have a stop under a child $r$ of $h$ only if $r$ is a partner of $v$.
\end{enumerate}
Indeed, a non-mediator outside gadget $u$ below $h$ would contradict (Q3). A mediator child $r$ of $h$ has $\mathrm{Loc}(r)=u$, so a stop of $v$ under $r$ is assigned to $u$ (Q2). A stop of $v$ under a non-mediator of gadget $u$ that is not a partner is excluded by (Q2).

\begin{proof}[Proof of \autoref{lem:np-rigid}]
We treat the connectors in the order $U_u$, $c^1_u$, $q_u$, $p_u$, $c^{\mathrm{in}}_u$, $w_u$, $\bar w_u$, $c^{\mathrm{out}}_u$. Each case uses only the cases before it, (O), and the partner lists from (F4): the partners of $p_u$ are $q_u,w_u$; of $q_u$ are $p_u,c^1_u$; of $c^1_u$ are $q_u,\bar w_u$; of $w_u$ is $p_u$; of $\bar w_u$ is $c^1_u$; and $U_u$, $c^{\mathrm{in}}_u$, $c^{\mathrm{out}}_u$ have none in gadget $u$.

\emph{$U_u$} ($T_u\to S_u$). The first move goes up to $\pa(T_u)=S_u=t(U_u)$, after which $U_u$ never moves again, or down into a child of $T_u$. The latter would be a stop under a child of $T_u$, which by (O) must be a partner, and $U_u$ has none.

\emph{$c^1_u$} ($U_u\to S_u$). $U_u$ moves once, from $T_u$ to $S_u$. The first move of $c^1_u$ goes up to $\pa(U_u)\in\{T_u,S_u\}$ or down into a child of $U_u$. A stop under $T_u$ is excluded by (Q2). $S_u$ is the target, so this is the only move. A child of $U_u$ must be a partner (O): $q_u$ or $\bar w_u$. But $\bar w_u$ is never below $U_u$, which is not its partner (Q2), and $q_u$ is a child of $U_u$ only at its target, after which it never moves again. From a stop under $q_u$, $c^1_u$ can go up to $\pa(q_u)=U_u=s(c^1_u)$, which is excluded, or down into a child of $q_u$, which by (O) must be a partner of $c^1_u$. The partners of $c^1_u$ are $q_u$ and $\bar w_u$, and $\bar w_u$ is never a child of $q_u$ (not its partner, Q2). So $c^1_u$ never reaches $S_u$: a contradiction.

\emph{$q_u$} ($T_u\to U_u$). Consider the last move. Either it goes down into $U_u$ from $\pa(U_u)\in\{T_u,S_u\}$: from $T_u=s(q_u)$ it is the only move, and a stop under $S_u$ is excluded. Or it goes up from a stop under a child of $U_u$, which by (O) is a partner of $q_u$; since $p_u$ is never below $U_u$ (not its partner), this is $c^1_u$. To enter $c^1_u$, $q_u$ comes down from $\pa(c^1_u)=U_u=t(q_u)$, which is excluded, or up from a child of $c^1_u$, which would be a partner of $q_u$, hence $p_u$. But $p_u$ is never below $c^1_u$ (not its partner). Contradiction.

\emph{$p_u$} ($S_u\to T_u$). The last move goes down into $T_u$ from $S_u=s(p_u)$, which makes it the only move, or up from a child of $T_u$ that is a partner: $q_u$, since $w_u$ is never below $T_u$ (not its partner). To enter $q_u$, $p_u$ comes down from $\pa(q_u)=T_u=t(p_u)$ (excluded) or up from a child of $q_u$ that is a partner of $p_u$. That partner would be $w_u$, which is never below $q_u$. Contradiction.

\emph{$c^{\mathrm{in}}_u$} ($0\to S_u$). The last move goes down into $S_u$ from $\pa(S_u)=0=s(c^{\mathrm{in}}_u)$, which makes it the only move, or up from a child of $S_u$, which by (O) would have to be a partner. There is none.

\emph{$w_u$} ($c^{\mathrm{in}}_u\to S_u$). The last move goes down into $S_u$ from $0$, which is a stop under $0$ and is excluded, or up from a child of $S_u$. That child is $c^{\mathrm{in}}_u=s(w_u)$, after $c^{\mathrm{in}}_u$ has moved, and then this is the only move of $w_u$. Or it is the partner $p_u$. To enter $p_u$, $w_u$ comes down from $\pa(p_u)=S_u=t(w_u)$ (excluded) or up from a child of $p_u$ that is a partner of $w_u$; the only partner is $p_u$ itself. Contradiction.

\emph{$\bar w_u$} ($S_u\to c^{\mathrm{out}}_u$). The first move goes up to $0$ (excluded) or down into a child of $S_u$. That child is $c^{\mathrm{out}}_u=t(\bar w_u)$, which makes it the only move, or the partner $c^1_u$, after $c^1_u$ has reached $S_u$, where it stays. From $c^1_u$, $\bar w_u$ can go up to $S_u=s(\bar w_u)$ (excluded) or down into a child of $c^1_u$ that is a partner of $\bar w_u$; there is none. Contradiction.

\emph{$c^{\mathrm{out}}_u$} ($S_u\to0$). The first move goes up to $0=t(c^{\mathrm{out}}_u)$, which makes it the only move, or down into a child of $S_u$, which would have to be a partner. There is none.
\end{proof}

In particular, a direct move realizes the orders of \autoref{tab:gadget}. For example, $w_u$ moves up from $c^{\mathrm{in}}_u$ to $S_u$, so $c^{\mathrm{in}}_u$ is already a child of $S_u$. And $\bar w_u$ moves down into $c^{\mathrm{out}}_u$, so $c^{\mathrm{out}}_u$ is still a child of $S_u$.

\begin{proof}[Proof of \autoref{lem:np-gate}]
$S_u$ never moves (Q1), and a move changes the parent of one mutation by one level. Hence a mutation enters or leaves $R_u$ only when it, or an ancestor $w$, moves down from $0$ into $S_u$ or up from $S_u$ to $0$. Then $w$ has $0$ and $S_u$ as consecutive parents.

\emph{$w$ is not a mediator.} If $w$ is outside gadget $u$, then $w$ is a child of $S_u$, which is excluded by (Q3). If $w$ is in gadget $u$, then by (Q1) and \autoref{lem:np-rigid} it moves once, directly. The only fixed node or connector of $u$ whose single move goes between $0$ and $S_u$ are $c^{\mathrm{in}}_u$ and $c^{\mathrm{out}}_u$.

\emph{$w$ is a mediator.} No mediator starts or ends at $0$ or $S_u$, so $w$ has a stop under $S_u$, assigned to $u$, and a stop under $0$. By (Q4), the stop under $S_u$ is the free stop of $w$. This means one of two things.
\begin{itemize}
\item It is the standard stop, so the port of $w$ is $S_u$, and $w$ is $\langle w_u,p_u\rangle$ or $\langle c^1_u,\bar w_u\rangle$.
\item $w$ has no standard stop and the stop under $S_u$ is its first stop. Then the first move of $w$ goes up from $s(w)$, a child of $S_u$, which lies in gadget $u$ by (Q3). So $w$ is an internal mediator of $u$ or some $x_{uc}$. But an $x_{uc}$ that stops under $0$ makes a standard stop, which is excluded here.
\end{itemize}
In both cases $w$ is an internal mediator of $u$, and its stop under $0$ is a paid stop assigned to $\varphi(w)=u$. This is a contradiction.
\end{proof}

\begin{proof}[Proof of \autoref{lem:np-internal}]
Let $y=\langle p,q\rangle$ be an internal mediator of $u$. By \autoref{lem:np-rigid}, $p$ and $q$ never have a parent outside $C(\cdot)$, so by (F1) they are never parent and child. Hence $y$ moves at least twice.

\emph{The first stop of $y$ is assigned to $u$.} At its first move, at time $\theta_1$, $y$ goes up to $\pa(p)\in C(p)$ or down into a child of $p$. The canonical parents of the sources of internal mediators are non-mediators of gadget $u$. A child of $p$ is not a non-mediator: canonically, there is none by (F3); a displaced one from gadget $u$ is excluded by rigidity; one from outside is excluded by (Q3), since $p\notin\{c^{\mathrm{in}}_u,c^{\mathrm{out}}_u\}$ by (F3). So it is a mediator, located in $u$ at $\theta_1$. In both cases the stop is assigned to $u$.

\emph{The last stop of $y$ is assigned to $u$.} We prove the following, by induction on $\theta$.

\smallskip
\noindent(D) \emph{Let $z$ be a child of a mediator $r$ from time $\theta_e$ until time $\theta$, where $r$ is a child of $q$ at time $\theta$, and $z=y$ or $z$ is an ancestor of $y$ at some moment in $[\theta_e,\theta]$. Then the stop of $z$ under $r$ is assigned to $u$.}
\smallskip

Every such $r$ is an ancestor of $y$ at some moment, so $r\ne y$. If $r$ does not move during $[\theta_e,\theta]$, then $r$ is a child of $q$ at $\theta_e$ and $\mathrm{Loc}(r)=u$. Otherwise, let $\theta_r\in(\theta_e,\theta]$ be the last move of $r$, which makes $r$ a child of $q$. The position of $r$ under $q$ is not its start (it moved) and not its end (by (F3) only $y$ has target $q$). So it is a stop of $r$, assigned to $u$. Before $\theta_r$, $r$ was a child of $\pa(q)$ (and moved down into $q$) or of a child $r'$ of $q$ (and moved up).
\begin{itemize}
\item In the first case, the position of $r$ under $\pa(q)\in C(q)$, a non-mediator of gadget $u$, is either its start or a stop. If it is a stop, it is assigned to $u$, so $r$ has two stops assigned to $u$, contradicting (Q4). If it is the start, then $r$ did not move before $\theta_r$, so at time $\theta_e$ it was a child of $s(r)$, a non-mediator of gadget $u$, and $\mathrm{Loc}(r)=u$.
\item In the second case, $r'$ is not a non-mediator (as in the previous paragraph, $q$ has no such children), and $r$ is not at its start under a mediator. By (D) for $z=r$ at time $\theta_r$ (whose own last move precedes $\theta_r$), the stop of $r$ under $r'$ is assigned to $u$. Again $r$ has two stops assigned to $u$: a contradiction.
\end{itemize}
At its last move $y$ goes up from a child of $q$ or down from $\pa(q)$ into $q$. In the second case the last stop is under $\pa(q)\in C(q)$, a non-mediator of gadget $u$, so it is assigned to $u$. In the first case it is under a mediator child of $q$, and (D) applies with $z=y$.

\emph{Conclusion.} By (Q4), the first and the last stop of $y$ coincide, so $y$ makes exactly one stop $f$, entered at time $\theta_1$ and left at time $\theta_2$. There are four cases.
\begin{itemize}
\item[(i)] $f=\pa(p)$ at $\theta_1$ and $\pa(q)=f$ at $\theta_2$. By rigidity $f\in C(p)\cap C(q)=\{t(p)\}$ (F1). So $p$ is already at $t(p)$ at time $\theta_1$, and $q$ is still at $s(q)$ at time $\theta_2$. Hence $p$ moves before $q$.
\item[(ii)] $f=\pa(p)$ at $\theta_1$ and $q=\pa(f)$ at $\theta_2$. Then $f\in C(p)$ is a fixed node or a connector of gadget $u$, and by rigidity its parent lies in $C(f)$, so $q\in C(f)$. This contradicts (F1).
\item[(iii)] $f$ is a child of $p$ at $\theta_1$ and of $q$ at $\theta_2$. As shown above, a child of $p$ is a mediator $f\ne y$. Its position under $p$ is not its start ($s(f)=p$ would give $f=y$ by (F3)) and not its end (it moves later). Its position under $q$ is not its end ($t(f)=q$ would give $f=y$) and not its start. So $f$ has two stops assigned to $u$, contradicting (Q4).
\item[(iv)] $f$ is a child of $p$ at $\theta_1$ and $\pa(q)=f$ at $\theta_2$. Then $f\in C(q)$ is a fixed node or connector of gadget $u$ whose parent at $\theta_1$ is $p$, so $p\in C(f)$ by rigidity. This contradicts (F1).
\end{itemize}
Only case (i) is possible.
\end{proof}

\end{document}